\documentclass[pra,showpacs,superscriptaddress,amssymb,twocolumn,longbibliography]{revtex4-1}
\usepackage{graphicx}
\usepackage{amsmath}
\usepackage{amsthm}
\usepackage{amssymb}
\usepackage[usenames]{color}
\usepackage{hyperref}
\hypersetup{
    colorlinks=true,       
    linkcolor=cyan,          
    citecolor=magenta,        
    filecolor=magenta,      
    urlcolor=cyan,           
    runcolor=cyan
}

\usepackage{bm}
\usepackage{threeparttable}
\usepackage{subfigure}
\usepackage{upgreek }
\usepackage{marginnote}
\def\wich#1{\mathinner{\langle{#1}\rangle}}
\newcommand{\beq}{\begin{equation}}
\newcommand{\eeq}{\end{equation}}
\newcommand{\beqnn}{\begin{equation*}}
\newcommand{\eeqnn}{\end{equation*}}
\newcommand{\bea}{\begin{eqnarray}}
\newcommand{\eea}{\end{eqnarray}}
\newcommand{\beann}{\begin{eqnarray*}}
\newcommand{\eeann}{\end{eqnarray*}}
\newcommand{\bes} {\begin{subequations}}
\newcommand{\ees} {\end{subequations}}

\newcommand{\ket}[1]{ | #1\rangle}
\newcommand{\bra}[1]{\langle #1 | }

\newcommand{\mO}{\mathcal{O}}

\newcommand{\ident}{\openone}
\newcommand{\Tr}{\mathrm{Tr}}
\newcommand{\eps}{\varepsilon}
\newcommand{\abs}[1]{\ensuremath{\left| #1 \right|}}

\def\Ket#1{\left|#1\right>}

\newcommand{\ignore}[1]{}

\newtheorem{lemma}{Lemma}

\begin{document}
\title{When 
Diabatic Trumps Adiabatic in Quantum Optimization}
\author{Siddharth Muthukrishnan}
\affiliation{Department of Physics and Astronomy, University of Southern California, Los Angeles, California 90089, USA}
\affiliation{Center for Quantum Information Science \& Technology, University of Southern California, Los Angeles, California 90089, USA}

\author{Tameem Albash}
\affiliation{Department of Physics and Astronomy, University of Southern California, Los Angeles, California 90089, USA}
\affiliation{Center for Quantum Information Science \& Technology, University of Southern California, Los Angeles, California 90089, USA}
\affiliation{Information Sciences Institute, University of Southern California, Marina del Rey, CA 90292}

\author{Daniel A. Lidar}
\affiliation{Department of Physics and Astronomy, University of Southern California, Los Angeles, California 90089, USA}
\affiliation{Center for Quantum Information Science \& Technology, University of Southern California, Los Angeles, California 90089, USA}
\affiliation{Department of Electrical Engineering, University of Southern California, Los Angeles, California 90089, USA}
\affiliation{Department of Chemistry, University of Southern California, Los Angeles, California 90089, USA}

\begin{abstract}
We provide and analyze examples that counter the widely made claim that tunneling is needed for a quantum speedup in optimization problems. The examples belong to the class of perturbed Hamming-weight optimization problems. In one case, featuring a plateau in the cost function in Hamming weight space, we find that the adiabatic dynamics that make tunneling possible, while superior to simulated annealing, result in a slowdown compared to a diabatic cascade of avoided level-crossings. This, in turn, inspires a classical spin vector dynamics algorithm that is at least as efficient for the plateau problem as the diabatic quantum algorithm. In a second case whose cost function is convex in Hamming weight space, the diabatic cascade results in a speedup relative to both tunneling and classical spin vector dynamics.
\end{abstract}
\maketitle
%
%
The possibility of a quantum speedup for finding the solution of classical optimization problems is tantalizing, as a quantum advantage for this class of problems would provide a wealth of new applications for quantum computing. The goal of many optimization problems can be formulated as finding an $n$-bit string $x_{\mathrm{opt}}$ that minimizes a given function $f(x)$, which can be interpreted as the energy of a classical 
Ising spin system, whose ground state is $x_{\mathrm{opt}}$.  Finding the ground state of such systems can be hard if, e.g., the system is strongly frustrated, resulting in a complex energy landscape that cannot be efficiently explored with any known algorithm due to the presence of many local minima \cite{Nishimori-book}. 
This can occur, e.g., in classical simulated annealing (SA) \cite{kirkpatrick_optimization_1983}, 
when the system's state is trapped in a local minimum.  
Provided the barriers between minima are sufficiently thin, quantum mechanics allows the system to tunnel in order to escape from such traps, though such comparisons must be treated with care since the quantum and classical potential landscapes are in general different.  It is with this potential advantage over classical annealing that quantum annealing \cite{PhysRevB.39.11828,finnila_quantum_1994,kadowaki_quantum_1998} and the quantum adiabatic algorithm (QA), were proposed \cite{farhi_quantum_2000}. Thermal hopping and quantum tunneling provide two starkly different mechanisms for solving optimization problems, and finding optimization problems that favor the latter continues to be an open theoretical question \cite{EPJ-ST:2015,Heim:2014jf}. To attack this question, in this work we focus on a well-known class of problems known as perturbed Hamming weight oracle (PHWO) problems. These are problems for which instances can be generated where QA either has an advantage over classical random search algorithms with local updates, such as SA \cite{Farhi-spike-problem,Reichardt:2004}, or has no advantage \cite{vanDam:01,Reichardt:2004}. Moreover, for PHWO problems with qubit permutation symmetry there is an elegant interpretation of tunneling in terms of a semiclassical potential \cite{Farhi-spike-problem,Schaller:2007uq}, which we exploit in this work.

If the total evolution time is sufficiently long so that the adiabatic condition is satisfied, QA is guaranteed to reach the ground state with high probability \cite{Jansen:07,lidar:102106}. However, this condition is only sufficient, and the scaling of the time to reach the adiabatic regime is therefore not necessarily the right computational complexity metric. The optimal time to solution (TTS$_\mathrm{opt}$), commonly used in benchmarking studies \cite{speedup} [also see the Supplementary Material (SM)], is a more natural metric. It is defined as the minimum total time such that the ground state is observed at least once with desired probability $p_d$:
\beq 
\label{eqt:TTSopt}
\text{TTS}_\mathrm{opt} = \min_{t_f > 0} \left( t_f \frac{\ln( 1- p_d)}{\ln \left(1-p_{\mathrm{GS}}(t_f) \right)} \right) \ ,
\eeq
where $t_f$ is the duration (in QA) or the number of single spin updates (in SA) of a single run of the algorithm, and $p_{\mathrm{GS}}(t_f)$ is the probability of finding the ground state in a single such run. The use of TTS$_\mathrm{opt}$ allows for the possibility that multiple short (diabatic) runs of the evolution, each lasting an optimal annealing time $({t_f})_{\textrm{opt}}$, result in a better scaling than a single long (adiabatic) run with an unoptimized $t_f$.

Here we demonstrate that for a specific class of PHWO problems, the optimal evolution time $t_f$ for QA is far from being adiabatic, and this optimal evolution involves no multi-qubit tunneling.  Instead the system leaves the ground state, only to return through a sequence of diabatic transitions associated with avoided-level crossings.  We also show that spin vector dynamics, which can be interpreted as a semi-classical limit of the quantum evolution with a product-state approximation, evolves in an almost identical manner. 

We note that PHWO problems are strictly toy problems since these problems are typically represented by highly non-local Hamiltonians (see 
the SM) and thus are not physically implementable, in the very same sense that the adiabatic Grover search problem is unphysical \cite{Roland:2002ul,RPL:10}. Nevertheless, these problems provide us with important insights into the mechanisms behind a quantum speed-up, or lack thereof.

\textit{The plateau problem}.---%
%
We focus on PHWO problems with a Hamiltonian of the form:
\beq 
\label{eq:pertham}
f(x) = \begin{cases} \abs{x} + p(\abs{x}) & l<\abs{x}<u,\\
\abs{x} & \text{elsewhere} \end{cases} \ , 
\eeq
where $\abs{x}$ denotes the Hamming weight of the bit string $x\in\{0,1\}^n$. 
The minimum of this function is 
$x_{\mathrm{opt}} = 00\cdots0$.
In the absence of a perturbation, i.e., $p(\abs{x}) = 0$, the problem can be solved efficiently classically using a local spin-flip algorithm such as SA, since flipping any `$1$' to a `$0$' will lower the energy. Note that SA can be interpreted as a random walk, where the decision to walk left or right and flip a spin is given by the Metropolis update rule. 

QA evolves the system from its ground state at $t=0$ according to a time-dependent Hamiltonian  
\beq 
\label{eqt:QuantumH}
H(s) =\frac{1}{2}\left(1-s \right)  \sum_i \left(\ident - \sigma_i^x \right)+ s \sum_{x} f(x) \ket{x} \bra{x} \ ,
\eeq
where we have chosen the standard transverse field ``driver'' Hamiltonian $H(0)$ that assumes no prior knowledge of the form of $f(x)$, and a linear interpolating schedule, with $s\equiv t/t_f$ being the dimensionless time parameter. 

\begin{figure}[b] 
   \subfigure[]{\includegraphics[width=0.32\textwidth]{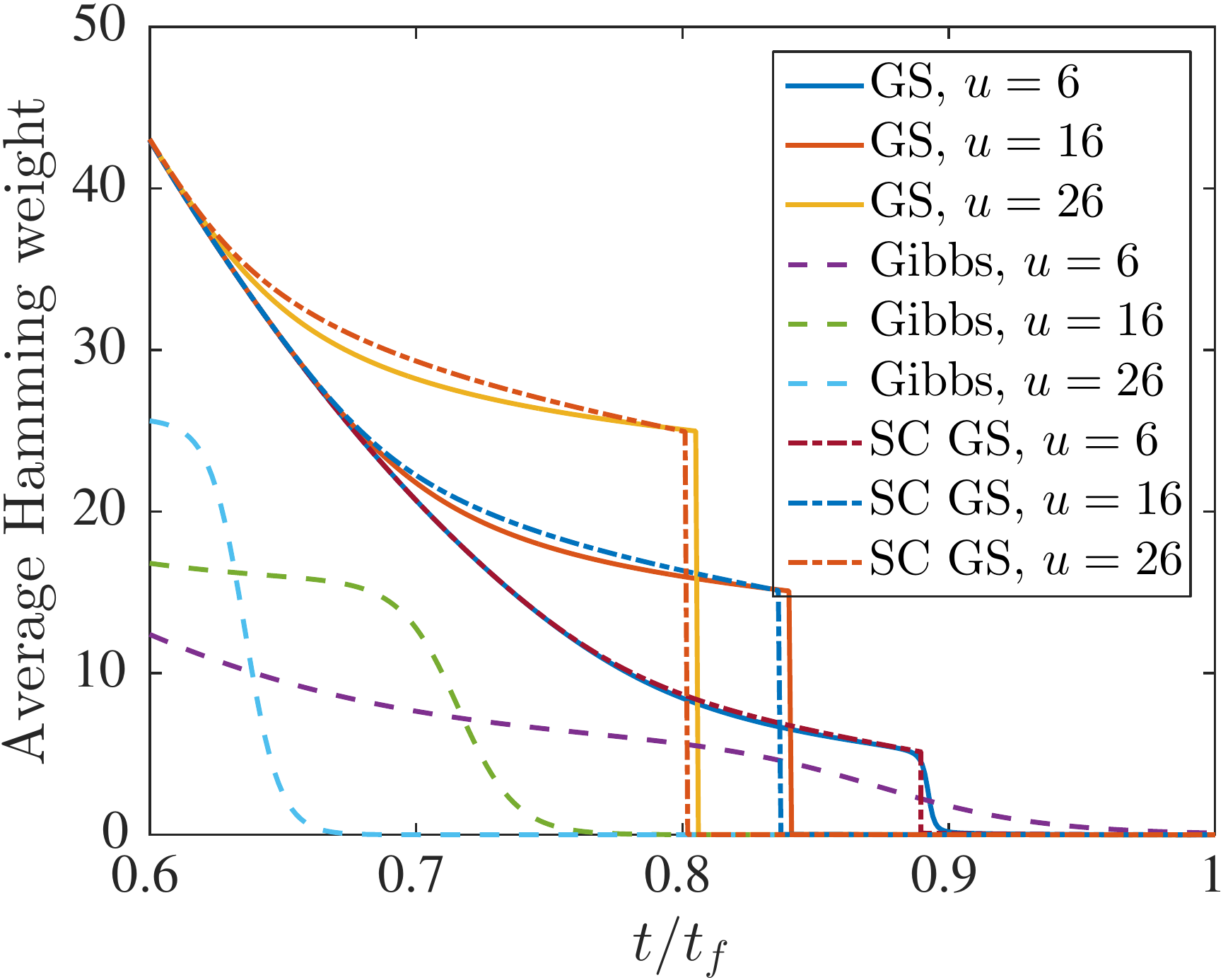} \label{fig:GSGibbs}}
   \subfigure[]{\includegraphics[width=0.32\textwidth]{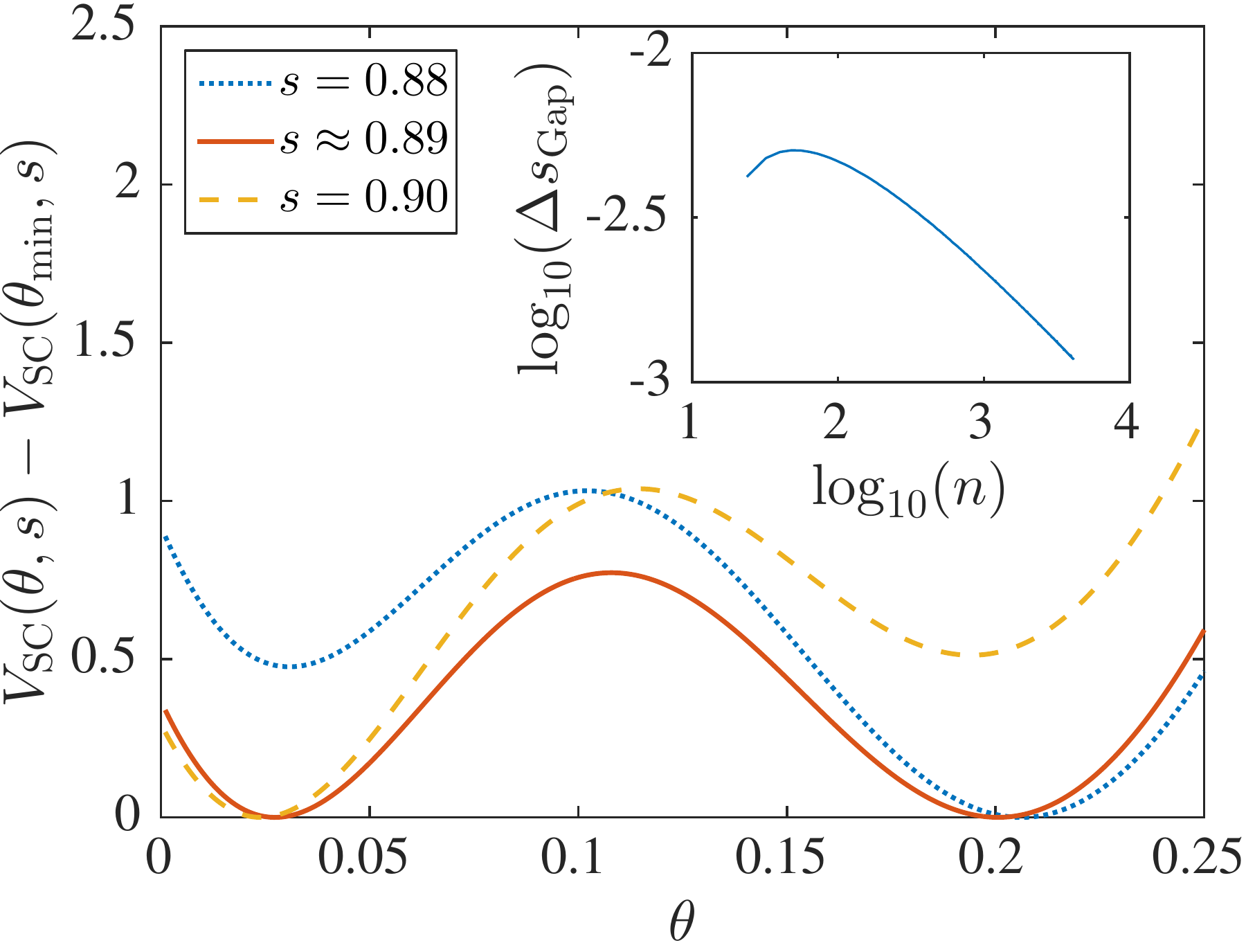} \label{fig:Veff}}
   \caption{(a) $\wich{\mathrm{HW}}$ in the instantaneous quantum ground state state (GS), the classical Gibbs state $\rho = e^{-\beta H_Z} / \mathcal{Z}$ (Gibbs), and the instantaneous quantum ground state predicted from the semi-classical potential (SC GS), as a function of their corresponding annealing parameters. The sharp drop in the GS and SC GS curves is due to a tunneling event wherein $\sim u$ qubits are flipped. Note that we use $t/t_f$ also for the Gibbs state, though in actuality the parameter is $\beta$, with a linear schedule: $\beta = 0.1 + 5.9s$.  (b) The semi-classical potential for $n=512$ and $u=6$ exhibits a double-well degeneracy at the position $s \approx 0.89$ (solid) of the sharp drop observed in (a), but is non-degenerate before and after this point (dotted and dashed), thus leading to a discontinuity in the position of its global minimum. The same is observed for other $u$ values we have checked (not shown). Inset: the difference in the position of the minimum gap from exact diagonalization and the position of the double-well degeneracy from the semi-classical potential, as a function of $n$, at $u=6$ (log-log scale).}
\end{figure}

Reichardt proved the following lower bound on the spectral gap for adiabatic evolutions for general PHWO problems \cite{Reichardt:2004}:
\beq \label{eqt:Reichbound}
\text{Gap}[H^{\text{pert}}(s)] \geq \text{Gap}[H^{\text{unpert}}(s)] - \mO\left(h \frac{u-l}{\sqrt{l}} \right),
\eeq
where $h = \max_x p(|x|)$ is the maximum height of the perturbation. 
Note that while the lower-bound holds for all perturbations, it is only non-trivial when it is positive for all $s \in [0,1]$. Details of our simulations methods 
and a summary of the proof are given in the SM.

We focus on the following ``plateau'' problem: 
%
\beq 
\label{eqt:plateau}
f(x) = \begin{cases}
u -1 &  l<\abs{x}<u,\\
\abs{x} & \text{otherwise}
\end{cases} \ .
\eeq
%
Thus, here $h=u-l-1$.  We note that when both $l,u = \mathcal{O}(1)$
a lower bound is not obtainable from Reichardt's proof (this is explained in 
the SM), although numerical diagonalization reveals a constant gap.  We demonstrate below that this case is nevertheless particularly hard for SA, and hence it is the focus of our study.  

\begin{figure*}[t] 
   \centering
 \subfigure[]{\includegraphics[width=0.32\textwidth]{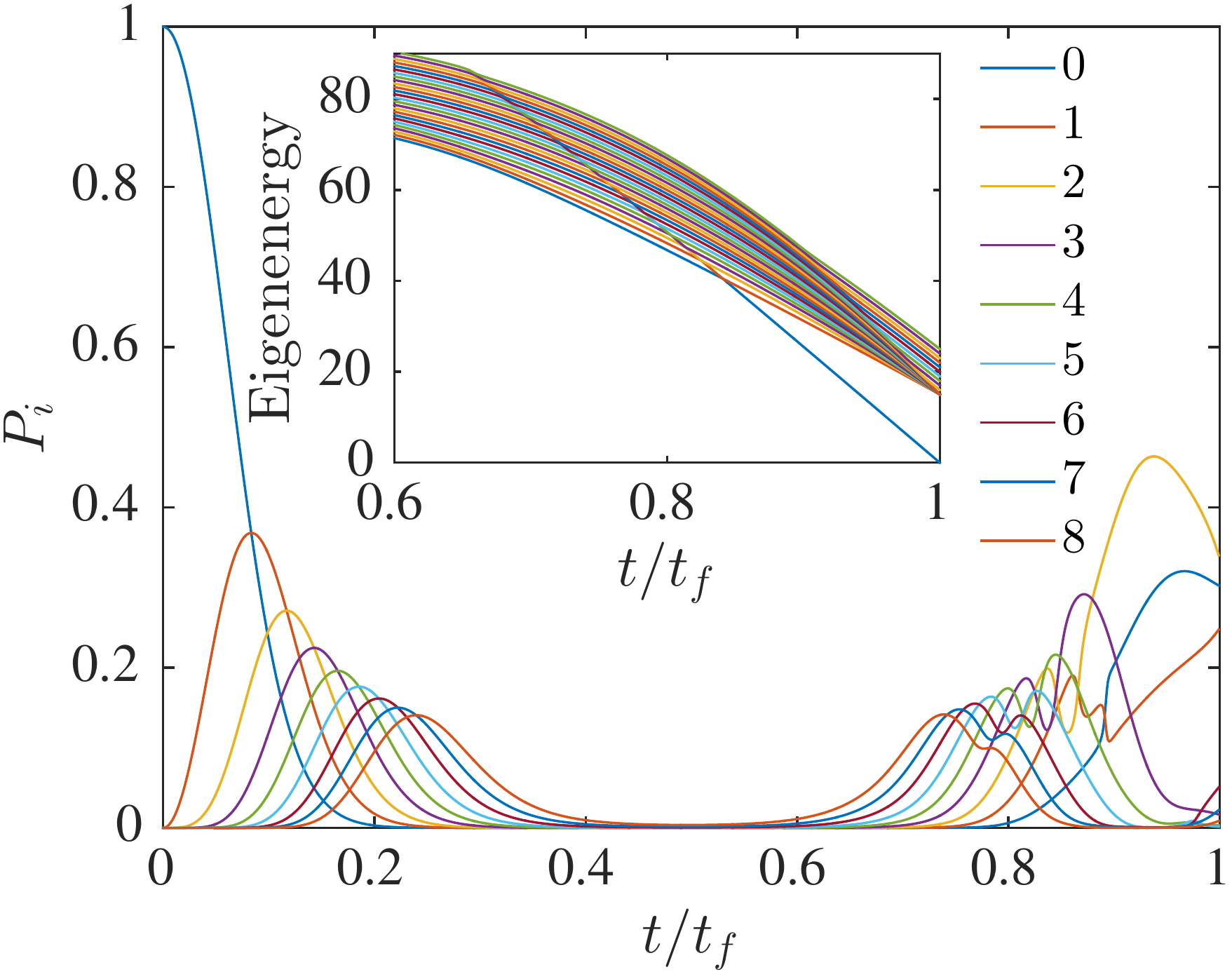}  \label{fig:QA_EnergyOverlap}}
   \subfigure[]{\includegraphics[width=0.32\textwidth]{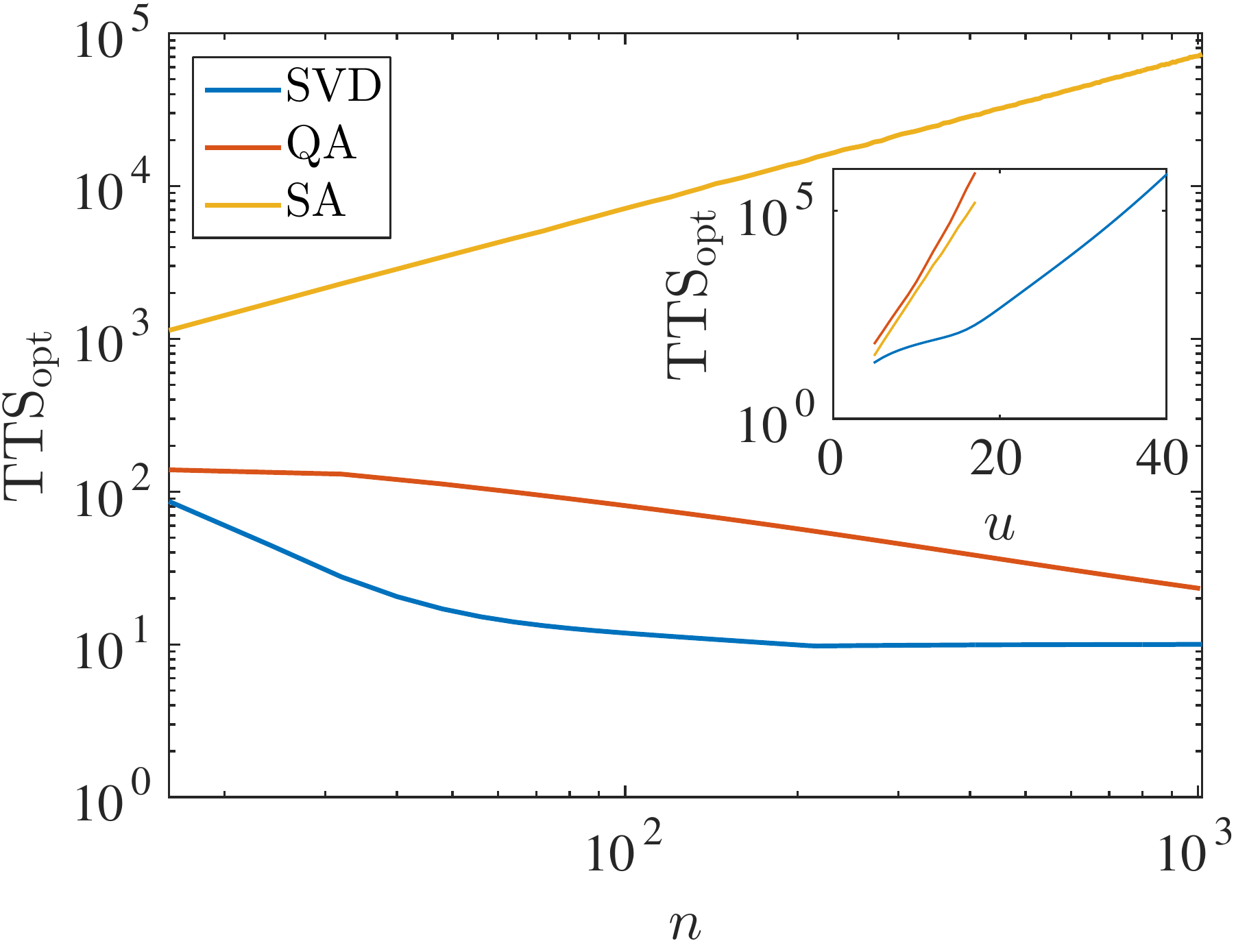} \label{fig:TTSScaling}}
  \subfigure[]{\includegraphics[width=0.32\textwidth]{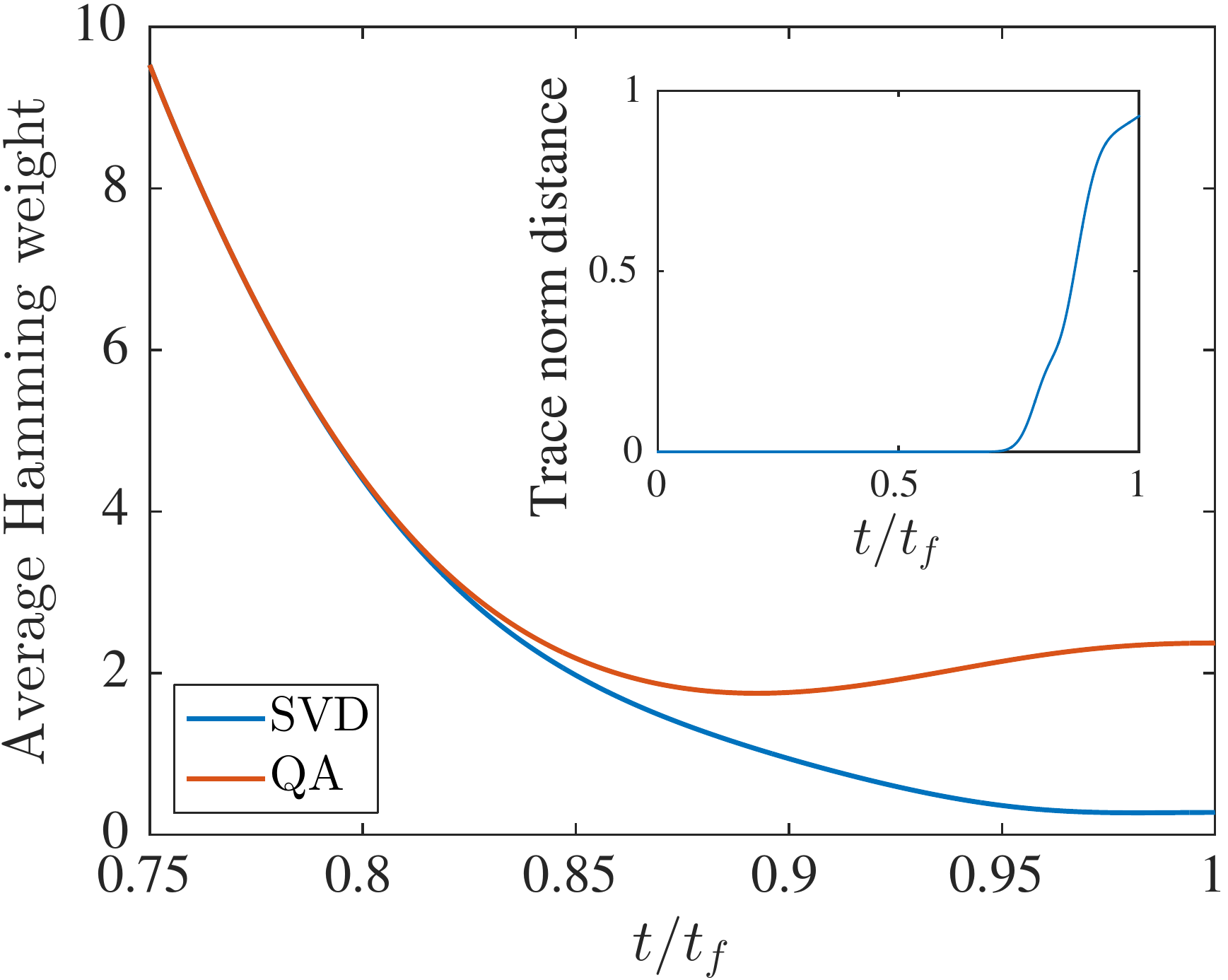}   \label{fig:QA_O3_AveHW}}
   \caption{Diabatic QA \textit{vs} SA and SVD. (a) Population $P_i$ in the $i$th energy eigenstate along the diabatic QA evolution at the optimal TTS for $n=512$ and $u=6$.  Excited states are quickly populated at the expense of the ground state.  By $t/t_f = 0.5$ the entire population is outside the lowest $9$ eigenstates.  In the second half of the evolution the energy eigenstates are repopulated in order.  Inset: the eigenenergy spectrum along the evolution.  Note the sequence of  avoided level crossings that unmistakably line up in the spectrum to reach the ground state. (b) Scaling of the optimal TTS with $n$ for $u=6$, with an optimized number of single-spin updates for SA, and equal $({t_f})_\textrm{opt}$ for QA and SVD. SA scales as $\mathcal{O}(n)$, a consequence of performing sequential single-spin updates. QA and SVD both approach $\mathcal{O}(1)$ scaling as $n$ increases. Here we set $p_d=0.7$ in Eq.~\eqref{eqt:TTSopt}, in order to be able to observe the saturation of SVD's TTS to the point where a single run suffices, i.e., TTS$_\textrm{opt} = ({t_f})_\textrm{opt}$. The conclusion is unchanged if we increase $p_d$: this moves the saturation point to larger $n$ for both SVD and QA, and we have checked that SVD always saturates before QA.
Inset: scaling as a function of $u$ for $n=1008$. SVD is again seen to exhibit the best scaling, while for this value of $n$ the scaling of QA and SA is similar (QA's scaling with $n$ improves faster than SA's as a function of $n$, at constant $u$). (c) $\wich{\mathrm{HW}}$ of the QA wavefunction and the SVD state (defined as the product of identical spin-coherent states) for $n=512$ and $u = 6$.  The behavior of the two is identical up to $t/t_f \approx 0.8$, when they begin to differ significantly, but neither displays any of the sharp changes observed in Fig.~\ref{fig:GSGibbs} for the instantaneous ground state.  Inset: the trace-norm distance between the QA and SVD states, showing that they remain almost indistinguishable until $t/t_f \approx 0.8$.}
\end{figure*}
%

\textit{Adiabatic dynamics}.---%
We now demonstrate that, under the assumption of adiabatic dynamics, QA efficiently tunnels through a barrier and solves the plateau problem in at most linear time, while SA requires a time that grows polynomially in the problem size $n$.

In the adiabatic (long-time) limit, SA follows the thermal Gibbs state parametrized by the inverse temperature $\beta$, whereas QA follows the instantaneous ground state parametrized by $s$. 
We quantify the distance of these two states from the target (the $\ket{0}^{\otimes n}$ state) by calculating the expectation value of the Hamming weight operator, defined as $\mathrm{HW} = \frac{1}{2} \sum_{i=1}^n \left( \ident - \sigma^z_i \right)$.  In particular, when $\wich{\mathrm{HW}} \approx 0$, the system has a high probability of being in the state $\ket{0}^{\otimes n}$.

As we tune (the annealing parameters) $\beta$ and $s$, the plateau in Eq.~\eqref{eqt:plateau} induces a dramatic change in $\wich{\mathrm{HW}}$, of the order of the plateau width $u$ and over a narrow range of the annealing parameters, as shown in Fig.~\ref{fig:GSGibbs}.  
For SA, traversing the plateau to reach the Gibbs state is a hard problem because, as the random walker moves closer to $x_{\mathrm{opt}}$, the probability to hop in the wrong direction of increasing Hamming weight becomes greater than the reverse. As we prove in 
the SM, consequently SA takes $\mathcal{O}(n^{u-l-1})$ single-spin updates to find the ground-state.

Unlike SA, to solve the plateau problem QA must tunnel through an energy barrier. To demonstrate this, we first construct the semiclassical effective potential arising from the spin-coherent path-integral formalism \cite{klauder1979path}:
\beq \label{eq:vsc}
V_{\mathrm{SC}} = \bra{\theta,\phi} H(s) \ket{\theta,\phi}
\eeq
where $\ket{\theta,\phi}$ are the spin-1/2 (symmetric) coherent states \footnote{%
The semiclassical potential has been used profitably in various QA studies, e.g., Refs.~\cite{Farhi-spike-problem,Schaller:2007uq,FarhiAQC:02,Boixo:2014yu}}.  
The potential captures all the important features of the quantum Hamiltonian~\eqref{eqt:QuantumH}: it displays a degenerate double well potential almost exactly at the point of the minimum gap [see Fig.~\ref{fig:Veff}]; the discontinuous change in the position of the global minimum of the potential gives rise to a nearly identical change in $\wich{\textrm{HW}}$ for the spin-coherent state [see Fig.~\ref{fig:GSGibbs}]. This agreement improves with increasing $n$, which is expected from  standard large-spin arguments \cite{Lieb:1973}.  
The dramatic drop in $\wich{\textrm{HW}}$ seen in Fig.~\ref{fig:GSGibbs} implies that $\sim u$ qubits have to be flipped in order to follow the instantaneous ground state.  Since the double-well potential becomes degenerate at the point where this flipping happens, as seen in Fig.~\ref{fig:Veff}, QA tunnels through the barrier in order to adiabatically follow the ground state. The constant minimum gap implies that this tunneling event happens in a time that is dictated by the scaling of the numerator of the adiabatic condition.
In our case this numerator turns out to be well approximated by the matrix element of $H(s)$ between the ground and first excited states, leading to $t_f\sim\mathcal{O}(n^{0.5})$ in the adiabatic limit (see 
the SM for details), thus polynomially outperforming SA whenever $u-l\geq 2$.

\textit{Optimal QA via Diabatic Transitions}.---%
Even though QA encounters a constant gap and can tunnel efficiently, the possibility remains that this does not lead to an optimal TTS, since this result assumes the adiabatic limit. We thus consider a diabatic form of QA and next demonstrate, using the optimal TTS criterion defined in Eq.~\eqref{eqt:TTSopt}, that the optimal annealing time for QA is far from adiabatic. Instead, as shown in Fig.~\ref{fig:QA_EnergyOverlap}, the optimal TTS for QA is such that the system leaves the instantaneous ground state for most of the evolution and only returns to the ground state towards the end. The cascade down to the ground state is mediated by a sequence of avoided energy level-crossings.  As $n$ increases for fixed $u$, repopulation of the ground state improves for fixed $({t_f})_{\textrm{opt}}$, hence causing TTS$_\mathrm{opt}$ to decrease with $n$, as seen Fig.~\ref{fig:TTSScaling}, until it saturates to a constant at the lowest possible value, corresponding to a single run at $({t_f})_{\textrm{opt}}$. At this point the problem is solved in constant time $({t_f})_{\textrm{opt}}$, compared to the $\sim\mathcal{O}(n^{0.5})$ scaling of the adiabatic regime. Moreover, as shown in Fig.~\ref{fig:QA_O3_AveHW}, there are no sharp changes in $\wich{\mathrm{HW}}$, suggesting that the non-adiabatic dynamics do not entail multi-qubit tunneling events, unlike the adiabatic case.  

Given the absence of tunneling in the time-optimal quantum evolution, we are motivated to consider a semiclassical limit of the evolution, particularly that of classical spin-vector dynamics (SVD), which we describe in detail in 
the SM. SVD can be derived 
as the saddle-point approximation to the path integral formulation of QA in the spin-coherent basis \cite{owerre2015macroscopic}.  The equations share the same qubit permutation symmetry as the quantum Hamiltonian, which significantly simplifies the computation of the oracle [i.e., the potential (and its derivatives), now given by Eq.~\eqref{eq:vsc}] assumed for SVD.
The SVD equations are equivalent to the Ehrenfest equations for the magnetization under the assumption that the density matrix is a product state, i.e., $\rho = \otimes_{i=1}^n \rho_{i}$, where $\rho_i$ denotes the state of the $i$th qubit.  

As we show in Fig.~\ref{fig:TTSScaling}, the scaling of SVD's optimal TTS also saturates to a constant time, i.e., $({t_f})_{\textrm{opt}}$. Moreover, it reaches this value earlier (as a function of problem size $n$) than QA, thus outperforming QA for small problem sizes, while for large enough $n$ both achieve $\mathcal{O}(1)$ scaling. As seen in the inset, SVD's advantage persists as a function of $u$ at constant $n$. 

The dynamics of QA in the non-adiabatic limit are well approximated by SVD until close to the end of the evolution, as shown in Fig.~\ref{fig:QA_O3_AveHW}: the trace-norm distance between the instantaneous states of QA and SVD is almost zero until $t/t_f \approx 0.8$, after which the states start to diverge. This suggests that SVD is able to replicate the QA dynamics up to this point, and only deviates because this makes it more successful at repopulating the ground state than QA.

\textit{Discussion}.---%
For the class of PHWO problems studied here, we have demonstrated that tunneling is 
not necessary to achieve the optimal TTS. Instead, the optimal trajectory uses diabatic transitions to first scatter completely out of the ground state and return via a sequence of avoided level crossings. 
This use of diabatic transitions is similar in spirit to those studied in Refs.~\cite{Somma:2012kx, crosson2014different, Hen2014,Steiger:2015fk}, but there are some important differences.
essentially, our PHWO findings hold for the standard formulation of QA, without any fine-tuning of the interpolation schedule or the Hamiltonian. However, while both the adiabatic and diabatic quantum algorithms outperform SA for the plateau problem, the faster quantum diabatic algorithm is not better than the classical SVD algorithm for this problem. 
These results extend beyond the plateau problem: as we show in the SM 
even the ``spike" problem studied in Ref.~\cite{Farhi-spike-problem}---which is in some sense the antithesis of the plateau problem since it features a sharp spike at a single Hamming weight---also exhibits the diabatic-beats-adiabatic phenomenon, indicating that tunneling is not required to efficiently solve the problem. Moreover, SVD is faster for this problem as well. 

\begin{figure}[t] 
   \centering
\includegraphics[width=0.32\textwidth]{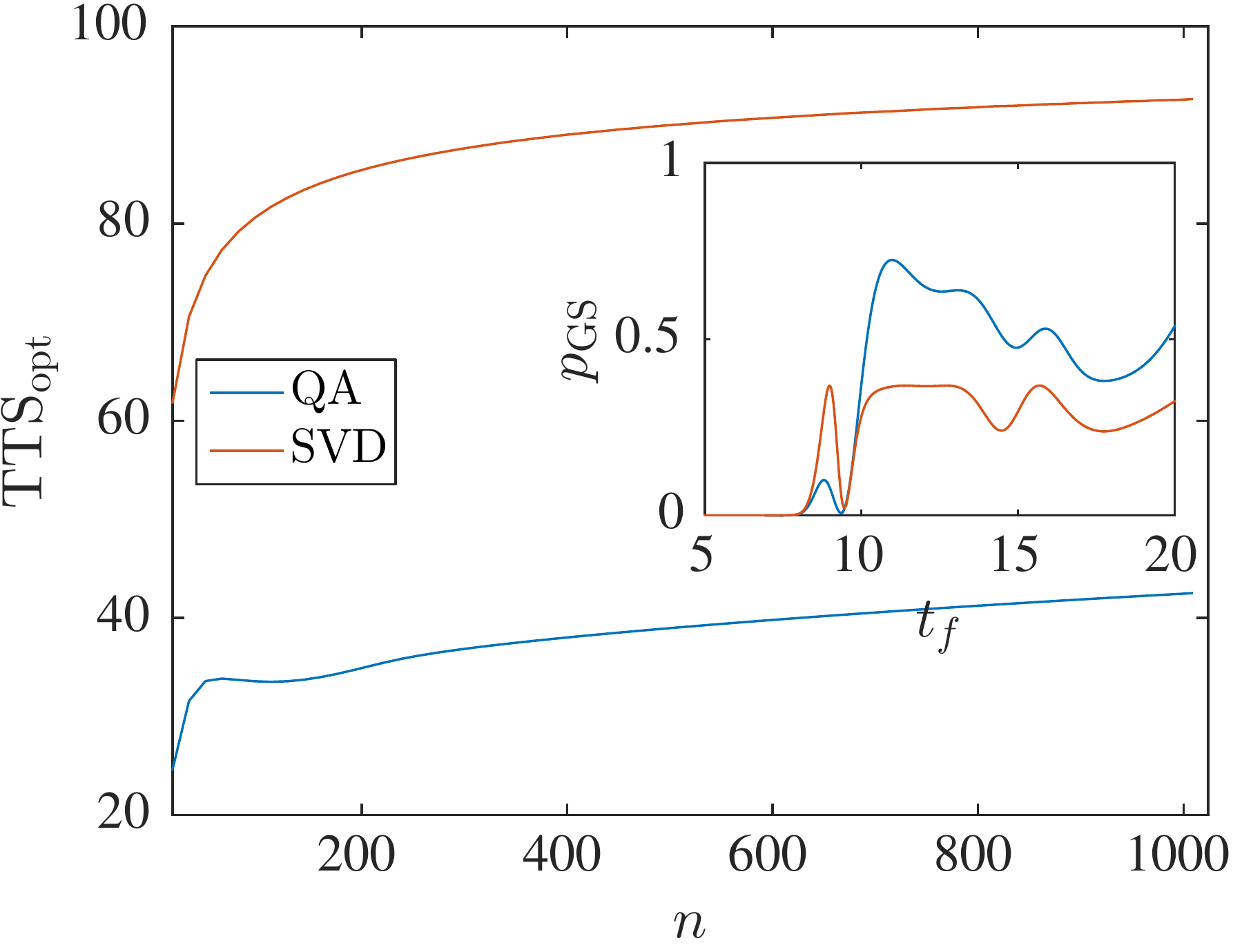}
   \caption{The optimal TTS for the potential given in Eq.~\eqref{eqt:jarret}.  QA outperforms SVD over the range of problem sizes we were able to check. The reason can be seen in the inset, which displays the ground state probability for SVD and QA for different annealing times $t_f$, with $n = 512$.  The optimal annealing time for SVD occurs at the first peak in its ground state probability ($t_f \approx 8.98$), whereas the optimal annealing time for QA occurs at the much larger \emph{second} peak in its ground state probability ($t_f \approx 10.91$).}
   \label{fig:ConvexPGS}
\end{figure}

However, the mechanism we found here, of a ``lining-up'' of the avoided level crossings with an associated ``diabatic cascade" [seen in Fig.~\ref{fig:QA_EnergyOverlap}], may not be generic. 
E.g., we have checked that this mechanism is absent in the adiabatic Grover problem with a transverse field driver Hamiltonian [as in our Eq.~\eqref{eqt:QuantumH}], even though the Grover problem is then equivalent to  a ``giant'' plateau problem: $f(x)=1-\delta_{x,0}$ \footnote{Note that this is not the version of the Grover problem that admits a quantum speedup, as this requires a rank-$1$ driver Hamiltonian \cite{Roland:2002ul}).}.

It is important to note that we do not claim that PHWO problems are always associated with diabatic cascades (see 
the SM for a counterexample); nor do we claim that SVD will always have an advantage over QA.  A simple counterexample to the latter statement comes from the class of cost functions that are convex in Hamming weight space, which have a constant minimum gap \cite{jarret2014fundamental}:
\beq 
\label{eqt:jarret}
f(x) = \begin{cases} 2, &  |x| =0 \\
\abs{x}, & \text{otherwise} \end{cases} \ .
\eeq
We have observed similar diabatic transitions for this problem as for the plateau (not shown), thus obviating tunneling, but find that QA outperforms SVD, as shown in Fig.~\ref{fig:ConvexPGS}.  This occurs because the optimum TTS for QA occurs at a slightly higher optimal annealing time, i.e., there is an advantage to evolving somewhat more slowly, though still far from adiabatically. Thus, this is a case of a ``limited'' quantum speedup \cite{speedup}.

In summary, our work provides a counterargument to the widely made claim that tunneling is needed for a quantum speedup in optimization problems. Which features of Hamiltonians of optimization problems favor diabatic or adiabatic algorithms remains an open question. 

\textit{Acknowledgements.}---%
Special thanks to Ben Reichardt for insightful conversations and for suggesting the plateau problem, and to Bill Kaminsky for inspiring talks~\cite{Kaminsky:2014,Kaminsky:USC-talk-2014}. We also thank Itay Hen, Joshua Job, Iman Marvian, Milad Marvian, and Rolando Somma for useful comments. The computing resources were provided by the USC Center for High Performance Computing and Communications and by the Oak Ridge Leadership Computing Facility at the Oak Ridge National Laboratory, which is supported by the Office of Science of the U.S. Department of Energy under Contract No. DE-AC05-00OR22725.  This work was supported under ARO grant number W911NF-12-1-0523 and ARO MURI Grant No. W911NF-11-1-0268.
%

\newpage

\onecolumngrid

\begin{center}
\textbf{\large{Supplementary Material for}}\\
\textbf{\large{``Diabatic Trumps Adiabatic in Quantum Optimization"}}\\
\end{center}
\twocolumngrid

\section{Derivation of E\lowercase{q}.~(1)}
Equation~\eqref{eqt:TTSopt} is easily derived as follows: the probability of successively failing $k$ times is $\left[1-p_{\mathrm{GS}}(t_f) \right]^k$, so the probability of succeeding at least once after $k$ runs is $1-\left[1-p_{\mathrm{GS}}(t_f) \right]^k$, which we set equal to the desired success probability $p_d$; from here one extracts the number of runs $k$ and multiplies by $t_f$ to get the time-to-solution TTS. Optimizing over $t_f$ yields TTS$_\mathrm{opt}$, which is natural for benchmarking purposes in the sense that it captures the trade-off between repeating the algorithm many times \textit{vs} optimizing the probability of success in a single run. The adiabatic regime might be more attractive if one seeks a theoretical guarantee to have a certain probability of success if the evolution is sufficiently slow. 

\section{(Non-)Locality of PHWO problems} \label{app:locality}
%

Since the PHWO problems, including the plateau, are quantum oracle problems, they can generically not be represented by a local Hamiltonian. For completeness we prove this claim here and also show why the (plain) Hamming weight problem is $1$-local.

Let $r$ be a bit string of length $n$, i.e., $r \in \{0,1\}^n$ and let 
\beq
\sigma^r \equiv \sigma_1^{r_1} \otimes \sigma_2^{r_2} \otimes \dots \otimes \sigma_n^{r_n},
\eeq
with $ \sigma_i^0 \equiv I_i$ and $\sigma_i^1 \equiv \sigma_i^z$.  This forms an orthonormal basis for the vector space of diagonal Hamiltonians. Thus:
\beq
H_P = \sum_{r \in \{0,1\}^n} J_r \sigma^r,
\eeq
with
\bes
\begin{align}
J_r &= \frac{1}{2^n}\Tr(\sigma^r H_P) \\
& = \frac{1}{2^n}\sum_{x \in \{0,1\}^n} f(x) \bra{x}\sigma^r \ket{x} \\
& = \frac{1}{2^n}\sum_{x \in \{0,1\}^n} f(x) (-1)^{x \cdot r}.
\end{align}
\ees
Note that generically $J_r$ will be be non-zero for arbitrary-weight strings $r$, leading to $|r|$-local terms in $H_P$, even as high as $n$-local.

E.g., substituting the plateau Hamiltonian [Eq.~\eqref{eqt:plateau}] into this we obtain:
\begin{eqnarray}
J_r &=& \frac{1}{2^n} \left[ \sum_{\abs{x} \leq l\, \&\, \abs{x} \geq u} \abs{x} (-1)^{x \cdot r} \right. \nonumber \\
&& \left. + (u-1)\sum_{l<\abs{x}<u} (-1)^{x \cdot r} \right].
\end{eqnarray}

On the other hand, if $f(x)=|x|$ (i.e., in the absence of a perturbation), the Hamiltonian is only $1$-local:
\bes
\begin{align}
H_P &= \sum_{x \in \{0,1\}^n} |x| \ket{x}\bra{x} \\
&= \sum_{x_1=0}^1 \dots \sum_{x_n =0}^1 (x_1 + x_2 + \dots + x_n) \ket{x_1}\bra{x_1} \nonumber \\
& \hspace{2cm} \otimes\ket{x_2}\bra{x_2} \otimes \dots \otimes \ket{x_n}\bra{x_n} \\
&= \sum_{k=1}^n \left(x_k \ket{x_k}\bra{x_k} \right) \bigotimes_{j\neq k} \left(\sum_{x_j=0}^1 \ket{x_j}\bra{x_j} \right) \\
&= \sum_{k=1}^n \ket{1}_k \bra{1} \bigotimes_{j\neq k} I_j = \sum_{k=1}^n \ket{1}_k \bra{1}.
\end{align}
\ees
%

\section{Methods} \label{app:Methods}
%

\subsection{Simulated Annealing}
\label{app:SA}
%
SA is a general heuristic solver \cite{kirkpatrick_optimization_1983}, whereby the system is initialized in a high temperature state, i.e., in a random state, and the temperature is slowly lowered while undergoing Monte Carlo dynamics.  Local updates are performed according to the Metropolis rule \cite{1953JChPh..21.1087M,HASTINGS01041970}: a spin is flipped and the change in energy $\Delta E$ associated with the spin flip is calculated.  The flip is accepted with probability $P_{\mathrm{Met}}$:
\beq
P_{\mathrm{Met}} = \min \{ 1 , \exp(-\beta \Delta E) \} \ ,
\eeq
where $\beta$ is the current inverse temperature along the anneal. Note that there could be different schemes governing which spin is to be selected for the update. We consider two such schemes: random spin-selection -- where the next spin to be updated is selected at random; and sequential spin-selection -- where one runs through all of the $n$ spins in a sequence. Random spin-selection (including just updating nearest neighbors) satisfies detailed-balance and thus is guaranteed to converge to the Boltzmann distribution.  Sequential spin-selection does not satisfy strict detailed balance (since the reverse move of sequentially updating in the reverse order never occurs), but it too converges to the Boltzmann distribution \cite{Manousiouthakis:1999}.  In sequential updating, a ``sweep" refers to all the spins having been updated once. In random spin-selection, we define a sweep as the total number of spin updates divided by the total number of spins. When it is possible to parallelize the spin updates, the appropriate metric of time-complexity is the number of sweeps $N_{\mathrm{SW}}$, not the number of spin updates (they differ by a factor of $n$) \cite{speedup}. However, in our problem this parallelization is not possible and hence the appropriate metric is the number of spin updates, and this is what is plotted in Fig.~\ref{fig:TTSScaling}.  After each sweep, the inverse temperature is incremented by an amount $\Delta \beta$ according to an annealing schedule, which we take to be linear, i.e. $\Delta \beta = ( \beta_f - \beta_i ) / (N_{\mathrm{SW}}-1)$.

We can use SA both as an \emph{annealer} and as a \emph{solver} \cite{Hen:2015rt}. In the former, the state at the end of the evolution is the output of the algorithm, and can be thought of as a method to sample from the Boltzmann distribution at a specified temperature. For the latter, we select the state with the lowest energy found along the entire anneal as the output of the algorithm, the better technique if one is only interested in finding the global minimum. We use the latter to maximize the performance of the algorithm.

More details concerning the performance of SA in the context of the problem studied here are presented below.
%
\subsection{Quantum Annealing}\label{app:methods_QD}

Here we consider the most common version of  quantum annealing:
\beq
H(s) = (1-s) \sum_{i=1}^n \frac{1}{2}(\ident_i - \sigma_i^x) + s\sum_{x \in \{0,1\}^n} f(x) \ket{x}\bra{x} \ ,
\eeq
where $s\equiv {t}/{t_f}$ is the dimensionless time parameter and $t_f$ is the total anneal time. The initial state is taken to be $\ket{+}^{\otimes n}$, which is the ground state of $H(0)$.

The initial ground state and the total Hamiltonian are symmetric under qubit permutations (recall that $f(x) = f(|x|)$ for our class of problems).  It then follows that the time-evolved state, at any point in time, will also obey the same symmetry. Therefore the evolution is restricted to the $(n+1)$-dimensional symmetric subspace, a fact that we can take advantage of in our numerical simulations.  This symmetric subspace is spanned by the Dicke states $\ket{S, M}$ with $S = n/2, M = -S, -S+1, \dots ,S$, which satisfy:
\bes
\begin{align}
S^2  \ket{S, M} &= S \left( S+1 \right) \ket{S, M}\\
S^z \ket{S, M} &= M \ket{S, M} \ ,
\end{align}
\ees
where $S^{x,y,z} \equiv \frac{1}{2} \sum_{i=1}^n \sigma_i^{x,y,z}$, $S^2 = \left(S^x \right)^2 + \left(S^y \right)^2 +\left(S^z \right)^2$.
We can denote these states by:
\beq \label{eqt:Dicke}
\ket{w} \equiv \Ket{\frac{n}{2} ,M = \frac{n}{2} - w} = {n \choose w}^{-1/2} \sum_{x: |x| = w} \ket{x}, 
\eeq
where, $w\in\{0,\dots,n\}$.

In this basis the Hamiltonian is tridiagonal, with the following matrix elements:
\bes \begin{align}
\left[H(s)\right]_{w,w+1} = & \left[H(s)\right]_{w+1,w} =\nonumber \\  &-\frac{1}{2} (1-s)  \sqrt{(n-w)(w+1)}, \\
\left[H(s)\right]_{w,w} = &(1-s) \frac{n}{2} + s  f(w).
\end{align} 
\ees
The Schr\"odinger equation with this Hamiltonian can be solved reliably using an adaptive Runge-Kutte Cash-Karp method \cite{Cash:1990:VOR:79505.79507} and the Dormand-Prince method \cite{Dormand198019} (both with orders $4$ and $5$).

If the quantum dynamics is run adiabatically the system remains close to the ground state during the evolution, and an appropriate version of the adiabatic theorem is satisfied. For evolutions with a constant spectral gap for all $s \in [0,1]$, an adiabatic condition of the form
\beq 
t_f \geq \text{const} \sup_{s \in [0,1]} \frac{ \|\partial_s H(s) \|}{\text{Gap}(s)^2}
\label{eq:trad-crit}
\eeq
is often claimed to be sufficient \cite{Messiah:vol2} (however, see discussion after Eq.~(21) in Ref.~\cite{Jansen:07}). In our case $\|\partial_s H(s) \| = \| H(1) - H(0) \|$ is upper-bounded by $n$;  since we are considering a constant gap, the adiabatic algorithm can scale at most linearly by condition \eqref{eq:trad-crit}. This is true for the plateau problems.

We demonstrate in SM-\ref{app:quantumscaling} that the following version of the adiabatic condition, known to hold in the absence of resonant transitions between energy levels \cite{Amin:09}, estimates the scaling we observe very well:
\beq
\max_{s \in [0,1]}  \frac{ |\bra{\eps_0(s) } \partial_s H(s) \ket{\eps_1(s)}|}{\text{Gap}(s)^2} \ll 1,
\label{eqt:AdC}
\eeq
where $\eps_0(s)$ and $\eps_1(s)$ are the instantaneous ground and excited states in the symmetric subspace respectively. To extract $t_f$ from this condition we simply ensure that for a given choice of $t_f$, condition~\eqref{eqt:AdC} holds (recall that $s=t/t_f$).

\subsection{Spin-Vector Dynamics }\label{app:methods_SVD}
%
Starting with the spin-coherent path integral formulation of the quantum dynamics, we can obtain Spin Vector Dynamics (SVD) as the saddle-point approximation (see, for example, p.10 of Ref.~\cite{owerre2015macroscopic} or Refs.~\cite{Smolin,Albash:2014if}). It can be interpreted as a semi-classical limit describing coherent single qubits interacting incoherently.  In this sense, SVD is a well motivated classical limit of the quantum evolution of QA.  SVD describes the evolution of $n$ unit-norm classical vectors under the Lagrangian (in units of $\hbar = 1$):
\beq
\mathcal{L} = i \wich{\Omega(s) | \frac{d}{ds} | \Omega(s)} - t_f \wich{\Omega(s) | H(s) | \Omega(s)},
\eeq
where $\ket{\Omega(s)}$ is a tensor product of $n$ independent spin-coherent states~\cite{arecchi1972atomic}:
\beq
\ket{\Omega(s)} = \bigotimes_{i=1}^n \left[ \cos\left( \frac{\theta_i(s)}{2} \right) \ket{0}_i + \sin \left( \frac{\theta_i(s)}{2} \right) e^{i \varphi_i(s)} \ket{1}_i \right].
\eeq
We can define an \emph{effective semiclassical potential} associated with this Lagrangian:
\begin{align}
\label{eqt:bigvsc}
&V_\mathrm{SC}(\{\theta_i\},\{\varphi_i\},s) \equiv \wich{\Omega(s) | H(s) | \Omega(s)} \nonumber \\
&=(1-s) \sum_{i=1}^n \frac{1}{2} \left(1- \cos\varphi_i(s) \sin\theta_i(s) \right) \nonumber\\
&+ s \sum_{x \in \{0,1\}^n} f(x) \prod_{j:x_j=0} \cos^2\left(\frac{\theta_j(s)}{2}\right)\prod_{j:x_j=1} \sin^2\left(\frac{\theta_j(s)}{2}\right),
\end{align}
with the probability of finding the all-zero state at the end of the evolution (which is the ground state in our case), as $\prod_{i=1}^n \cos^2\left(\frac{\theta_i(1)}{2}\right)$.
The quantum Hamiltonian obeys qubit permutation symmetry: $P H P = H$ where $P$ is a unitary operator that performs an arbitrary permutation of the qubits. This implies that our classical Lagrangian obeys the same symmetry:
\begin{eqnarray}
\mathcal{L}' &\equiv&  i \bra{\Omega(s)} P \frac{d}{ds} P \ket{\Omega(s)} - t_f \bra{\Omega(s)} P H(s) P \ket{\Omega(s)} \nonumber \\
&=&   i \bra{\Omega(s)}  \frac{d}{ds} \ket{\Omega(s)} - t_f \bra{\Omega(s)}  H(s) \ket{\Omega(s)} = \mathcal{L},  \nonumber\\
\end{eqnarray}
where the derivative operator is trivially permutation symmetric. Therefore, the Euler-Lagrange equations of motion derived from this action will be identical for all spins. Thus, if we have symmetric initial conditions, i.e., $(\theta_i(0),\varphi_i(0)) = (\theta_j(0),\varphi_j(0)) \ \forall i,j$, then the time evolved state will also be symmetric:
\beq
(\theta_i(s),\varphi_i(s)) = (\theta_j(s),\varphi_j(s)) \ \forall i,j \ \forall s\in[0,1] \ .
\eeq 
As we show below, under the assumption of a permutation-invariant initial condition we only need to solve two (instead of $2n$) semiclassical equations of motion:
\bes 
\label{eqt:symsvdeoms} 
\begin{align}
\frac{n}{2} \sin\theta(s) \theta'(s) - t_f  \partial_{\varphi(s)} V_\mathrm{SC}^\mathrm{sym} (\theta(s),\varphi(s),s)  = 0  \ , \\
- \frac{n}{2} \sin \theta(s) \varphi'(s) - t_f  \partial_{\theta(s)} V_\mathrm{SC}^\mathrm{sym} (\theta(s),\varphi(s),s) = 0 \ ,
\end{align} \ees
where we have defined the symmetric effective potential $V_\mathrm{SC}^\mathrm{sym}$ as:
\begin{align}
\label{eqt:symvsc}
&V_\mathrm{SC}^\mathrm{sym} (\theta(s),\varphi(s),s) \equiv \wich{\Omega^\mathrm{sym}(s) | H(s) | \Omega^\mathrm{sym}(s)} \nonumber \\
&=(1-s) \frac{n}{2} \left(1- \cos\varphi(s) \sin\theta(s) \right) \nonumber \\
 &+ s \sum_{w=0}^n f(w) \binom{n}{w} \sin^{2w}\left(\frac{\theta(s)}{2}\right) \cos^{2(n-w)}\left(\frac{\theta(s)}{2}\right),
\end{align}
and $\ket{\Omega^\mathrm{sym}(s)}$ is simply $\ket{\Omega(s)}$ with all the $\theta$'s and $\varphi$'s set equal. Note that in the main text [see Eq.~\eqref{eq:vsc}], we slightly abuse notation for simplicity, and use $V_\mathrm{SC}$ instead of $V_\mathrm{SC}^\mathrm{sym}$.
The probability of finding the all-zero bit string at the end of the evolution is accordingly given by $\cos^{2n}(\theta(1)/2)$. We would have arrived at the same equations of motion had we used the symmetric spin coherent state in our path integral derivation, but that would have been an artificial restriction. In our present derivation the symmetry of the dynamics naturally imposes this restriction.

Note that the object in Eq.~\eqref{eqt:bigvsc} involves a sum over all $2^n$ bit-strings and is thus exponentially hard to compute; on the other hand, the object in Eq.~\eqref{eqt:symvsc} only involves a sum over $n$ terms and is thus easy to compute. Therefore, just as in the quantum case---where due to permutation symmetry the quantum evolution is restricted to the $n+1$ dimensional subspace of symmetric states instead of the full $2^n$-dimensional Hilbert space---given knowledge of the symmetry of the problem we can efficiently compute the SVD potential and efficiently solve the SVD equations of motion. 

We also remark that the computation of the potential in Eq.~\eqref{eqt:bigvsc} is significantly simplified if our cost function, $f(x)$, is given in terms of a local Hamiltonian. For example, if $H(1) = \sum_{i,j} J_{ij} \sigma_i^z \sigma_j^z$, then:
\beq
V_\mathrm{SC}(\{\theta_i\},\{\varphi_i\},1) = \sum_{i,j} J_{ij} \cos\theta_i \cos\theta_j \ ,
\eeq
which is easy to compute if is a $\mathrm{poly}(n)$ number of terms.

Let us now derive the symmetric SVD equations of motion~\eqref{eqt:symsvdeoms}. Without any restriction to symmetric spin-coherent states, the SVD equations of motion, for the pair $\theta_i,\varphi_i$, read:
\bes 
\label{eqt:symsvdeoms2} 
\begin{align}
\frac{1}{2} \sin\theta_i(s) \theta_i'(s) - t_f  \partial_{\varphi_i(s)} V_\mathrm{SC} (\{\theta_i\},\{\varphi_i\},s)  = 0 \ , \\
- \frac{1}{2} \sin \theta_i(s) \varphi_i'(s) - t_f  \partial_{\theta_i(s)} V_\mathrm{SC}(\{\theta_i\},\{\varphi_i\},s) = 0 \ .
\end{align} 
\ees
As can be seen by comparing Eqs.~\eqref{eqt:symsvdeoms} and \eqref{eqt:symsvdeoms2}, it is sufficient to show that: 
\beq
\frac{\partial}{\partial \theta_i} V_\mathrm{SC} \bigg|_{\theta_j = \theta, \varphi_j = \varphi \ \forall j} = \frac{1}{n} \frac{\partial}{\partial \theta} V_\mathrm{SC}^\mathrm{sym},
\eeq
and an analogous statement holding for derivatives with respect to $\varphi$. This claim is easily seen to hold true for the term multiplying $(1-s)$ in Eq.~\eqref{eqt:bigvsc}:
\begin{align}
& \frac{\partial}{\partial \theta_i} \sum_{i=1}^n \frac{1}{2} \left(1- \cos\varphi_i(s) \sin\theta_i(s) \right)\bigg|_{\theta_j = \theta, \varphi_j = \varphi \ \forall j} \notag \\
& = \frac{\partial}{\partial \theta} \frac{1}{2} \left(1- \cos\varphi(s) \sin\theta(s) \right) \notag \\
& = \frac{1}{n} \frac{\partial}{\partial \theta}  V_\mathrm{SC}^\mathrm{sym} (\theta,\phi,s=0)\ ,
\end{align}
where in the last line we used Eq.~\eqref{eqt:symvsc}. Next we focus on the term multiplying $s$ in Eq.~\eqref{eqt:bigvsc}. This term has no $\varphi$ dependence and thus we only consider the $\theta$ derivatives. First note that
\begin{align}
&\frac{\partial}{\partial \theta_i} V_\mathrm{SC}(\{\theta_i\}, \{\varphi_i\},s=1) = \nonumber \\
&\sum_{x\in \{0,1\}^n } f(x) \prod_{j:x_j=0} \cos^2\left(\frac{\theta_j}{2}\right)\prod_{j:x_j=1} \sin^2\left(\frac{\theta_j}{2}\right) \nonumber \\
&\times \left[ -\delta_{x_i,0} \sec^2\left(\frac{\theta_i}{2}\right) + \delta_{x_i,1}\csc^2\left(\frac{\theta_i}{2}\right) \right] \frac{\sin\theta_i}{2}.
\end{align}
Now, we set all the $\theta_i$'s equal. Let us define $p(\theta) \equiv \sin^2\left(\frac{\theta}{2}\right)$. Using this and the fact that $f$ is only a function of the Hamming weight (which is equivalent to the qubit permutation symmetry), we can rewrite the last expression, after a few steps of algebra, as:
\begin{align}
\label{eq:32}
&\sum_{w=0}^n f(w) p^{w-1} (1-p)^{n-w-1} \partial_\theta p \nonumber \\
& \hspace{3cm} \times \left[ (1-p)\binom{n-1}{w-1} - p \binom{n-1}{w} \right] \nonumber \\
&= \sum_{w=0}^n f(w) p^{w-1} (1-p)^{n-w-1} \partial_\theta p  \left[ \frac{1}{n} \binom{n}{w} (w - np) \right] \notag \\
& = \frac{1}{n} \frac{\partial}{\partial \theta} V_\mathrm{SC}^\mathrm{sym}(\theta, \varphi, s=1)\ .
\end{align}
%

\subsection{Simulated Quantum Annealing}
Although not discussed in the main text, in SM-\ref{app:quantumscaling} we use an alternative method to simulated annealing, namely simulated quantum annealing (SQA, or Path Integral Monte Carlo along the Quantum Annealing schedule) \cite{sqa1,Santoro}. This is an annealing algorithm based on discrete-time path-integral quantum Monte Carlo simulations of the transverse field Ising model using Monte Carlo dynamics.  At a given time $t$ along the anneal, the Monte Carlo dynamics samples from the Gibbs distribution defined by the action:
\beq
S[\mu] =  \Delta(t) \sum_\tau H_{\mathrm{P}} (\mu_{:,\tau}) -J_{\perp}(t)  \sum_{i,\tau} \mu_{i,\tau} \mu_{i,\tau+1} 
\eeq
where $\Delta(t) = \beta B(t)/ N_{\tau}$ is the spacing along the time-like direction, $J_{\perp} = -0.5 \ln (\tanh(A(t)/2))$ is the ferromagnetic spin-spin coupling along the time-like direction, and $\mu$ denotes a spin configuration with a space-like direction (the original problem direction, indexed by $i$) and a time-like direction (indexed by $\tau$).
For our spin updates, we perform Wolff cluster updates \cite{PhysRevLett.62.361} along the imaginary-time direction only.  For each space-like slice, a random spin along the time-like direction is picked. The neighbors of this spin are added to the cluster (assuming they are
parallel) with probability 
\beq
P = 1 - \exp(-2 J_{\perp})
\eeq
When all neighbors of the spin have been checked, the newly added spins are checked.  When all spins in the cluster have had their neighbors along the time-like direction tested, the cluster is flipped according to the Metropolis probability using the space-like change in energy associated with flipping the cluster.  A single sweep involves attempting to update a single cluster on each space-like slice.

We can use SQA both as an \emph{annealer} and as a \emph{solver} \cite{Hen:2015rt}.  In the former, we randomly pick one of the states on the Trotter slices at the end of the evolution as the output of the algorithm, while for the latter, we pick the state with the lowest energy found along the entire anneal as the output of the algorithm.  We use the latter to maximize the performance of the algorithm.

\section{Review of the Hamming weight problem and Reichardt's bound for PHWO problems} \label{app:review}
Here we closely follow Ref.~\cite{Reichardt:2004}.

\subsection{The Hamming weight problem}

We review the analysis within QA of the minimization of the Hamming weight function 
$f_\mathrm{HW} (x)=\abs{x}$, which counts the number of $1$'s in the bit string $x$. This problem is of course trivial, and the analysis given here is done in preparation for the perturbed problem.

For the adiabatic algorithm, we start with the driver Hamiltonian,
\beq \label{eq:HD}
H_D = \frac{1}{2} \sum_{i=1}^n \left( \ident_i - \sigma^x_i \right) = \sum_{i=1}^n \ket{-}_i\bra{-} \ ,
\eeq
which has $\ket{+}^{\otimes n}$ as the ground state.

The final Hamiltonian for the cost function $f_\mathrm{HW} (x)$ is
\beq
H_P = \frac{1}{2} \sum_{i=1}^n \left( \ident_i - \sigma^z_i \right) = \sum_{i=1}^n \ket{1}_i\bra{1} \ ,
\eeq
which  has $\ket{0}^{\otimes n}$ as the ground state.

We interpolate linearly between $H_D$ and $H_P$:
\begin{align}
H(s) &= (1-s)H_D + sH_P; \quad s\in[0,1] \\
&= \sum_{i=1}^n \frac{1}{2}\begin{pmatrix}  1-s & -(1-s) \\ -(1-s) & 1-s\end{pmatrix}_i + \begin{pmatrix}  0 & 0 \\ 0 & s\end{pmatrix}_i,\\
&= \sum_{i=1}^n \frac{1}{2}\begin{pmatrix}  1-s & -(1-s) \\ -(1-s) & 1+s\end{pmatrix}_i \ .
\end{align}
Since there are no interactions between the qubits, this problem can be solved exactly by diagonalizing the Hamiltonian on each qubit separately. For each term, we have the energy eigenvalues $E_{\pm}(s)$,
\beq
E_{\pm}(s) = \frac{1}{2} (1 \pm \Delta(s)); \quad \Delta(s) \equiv \sqrt{1-2s+2s^2},
\eeq
and associated eigenvectors,
\beq
\ket{v_{\pm}(s)} = \frac{1}{\sqrt{2\Delta (\Delta \mp s)}} \left[ \mp(\Delta \mp s) \ket{0} + (1-s) \ket{1} \right] \ .
\eeq
The ground state of $H(s)$ is $\ket{v_{-}(s)}^{\otimes n}$. The gap is given by,
\bes
\begin{align}
\text{Gap}[H(s)] &= H(s) \ket{v_{+}(s)}\otimes\ket{v_{-}(s)}^{\otimes (n-1)} \nonumber \\
& \hspace{3cm} - H(s)\ket{v_{-}(s)}^{\otimes n} \\
&= E_+ + (n-1)E_- - nE_- \\
&= E_+ - E_- \\
&=\Delta(s) \ .
\end{align}
\ees
The gap is minimized at $s=\frac{1}{2}$ with minimum value $\Delta(\frac{1}{2})=\frac{1}{\sqrt{2}}$. The minimum gap is independent of $n$ and hence does not scale with problem size. Therefore the adiabatic run time is given by,
\beq
t_f = \mathcal{O}\left( \frac{\|H\|}{\Delta^2} \right)= \mathcal{O}(n)\ ,
\eeq
where the $n$-dependence is solely due to $\| H\|$ (see SM-\ref{app:methods_QD}).

It also useful to consider the form of $\ket{v_{-}(s)}^{\otimes n}$. We can write,
\begin{eqnarray}
\ket{v_{-}(s)}^{\otimes n} &=& \frac{1}{[2\Delta (\Delta + s)]^{\frac{n}{2}}} \times \nonumber \\
&& \hspace{-0.75cm} \sum_{x \in \{0,1\}^n} (1-s)^{\abs{x}} (\Delta(s) + s)^{n-\abs{x}} \ket{x} \ .
\end{eqnarray}
If we measure in the computational basis, the probability of getting outcome $x$ is determined by $\abs{x}$:
\beq
\text{Pr}[x](s) = \abs{\wich{v_{-}^{\otimes n}|x}}^2 = q(s)^{\abs{x}} (1-q(s))^{n-\abs{x}} \ ,
\label{eq:Pr[x](s)}
\eeq
where
\beq
q(s) \equiv \frac{(1-s)^2}{[2\Delta (\Delta + s)]} \ .
\eeq

\subsection{Reichardt's bound for PHWO problems} 
\label{app:ReichardtBound}
Here we review Reichardt's derivation of the gap lower-bound for general PHWO  problems, but provide additional details not found in the original proof~\cite{Reichardt:2004}. 

We use the same initial Hamiltonian [Eq.~\eqref{eq:HD}] and linear interpolation schedule as before, $\tilde{H}(s) = (1-s) H_D + s \tilde{H}_P$, and choose the final Hamiltonian to be
\beq 
\tilde{H}_P = \sum_{x\in \{0,1\}^n} \tilde{f} (x) \ket{x}\bra{x}\ ,
\eeq
where
\beq
\tilde{f}(x) = \begin{cases} \abs{x} + p(x) & l<\abs{x}<u \ ,\\
\abs{x} & \text{elsewhere} \end{cases} \ , 
\eeq
where $p(x) \geq 0$ is the perturbation.  Note that here we have \emph{not} assumed that the perturbation, $p(x)$, respects qubit permutation symmetry. 

We wish to bound the minimum gap of $\tilde{H}(s)$. Unlike the Hamming weight problem $H(s)$, this problem is no longer non-interacting. Define
\beq
h_k \equiv \max_{\abs{x}=k} p(x); \quad h \equiv \max_k h_k = \max_x p(x).
\eeq

\begin{lemma}[\cite{Reichardt:2004}]
\label{lem:bound}Let $u = \mathcal{O}(l)$ and let $E_0(s)$ and $\tilde{E}_0(s)$ be the ground state energies of $H(s)$ and $\tilde{H}(s)$, respectively. Then $\tilde{E}_0 (s) \leq E_0(s) + \mathcal{O}( h \frac{u-l}{\sqrt{l}})$.
\end{lemma}

\begin{proof}
First note that
\beq
\tilde{H}(s)- H(s) = s \sum_{x: l<\abs{x}<u} p(x) \ket{x}\bra{x}\ .
\eeq

Below, we suppress the $s$ dependence of all the terms for notational simplicity.  We know that $E_0 = \wich{v_{-}^{\otimes n}|H|v_{-}^{\otimes n}}$. Using this,
\bes
\begin{align}
\wich{\tilde{E}_0|\tilde{H}|\tilde{E}_0} &\leq \wich{\psi | \tilde{H} | \psi} \quad \forall \ket{\psi}\in\mathcal{H}. \\
\implies \tilde{E}_0 - E_0  &\leq \wich{v_{-}^{\otimes n}|\tilde{H}|v_{-}^{\otimes n}} - E_0 \\
&\leq \wich{v_{-}^{\otimes n}|\tilde{H}-H|v_{-}^{\otimes n}} \label{eq:perttheory}\\
&=s \sum_{x: l<\abs{x}<u} p(x) \abs{\wich{v_{-}^{\otimes n}|x}}^2 \\
&=s \sum_{x: l<\abs{x}<u} p(x) q^{\abs{x}}(1-q)^{n-\abs{x}} \\
&\leq \sum_{k:l<k<u} h_k {n \choose k}  q^k (1-q)^{n-k}, \label{eq:hk}
\end{align}
\ees
where ${n \choose k}$ is the number of strings with Hamming weight $k$, and we used Eq.~\eqref{eq:Pr[x](s)}. 

Consider the partial binomial sum (dropping the $h_k$'s),
\beq
\sum_{k:l<k<u} {n \choose k}  q^k (1-q)^{n-k}.
\eeq
Using the fact that the binomial is well-approximated by the Gaussian in the large $n$ limit (note that this approximation requires that $q(s)$ and $1-q(s)$ not be too close to zero), we can write:
\begin{align}
\label{eq:integralbound}
& \sum_{k:l<k<u} {n \choose k}  q^k (1-q)^{n-k}  \approx \int_l^u d\xi \ \frac{1}{\sqrt{2 \pi}\sigma} e^{-\frac{(\xi - \mu)^2}{2 \sigma^2}} \notag \\
&\quad = \frac{1}{\sigma} \int_l^u d\xi\ \phi\left(\frac{\xi-\mu}{\sigma}\right) = \int_{(l-\mu)/\sigma}^{(u-\mu)/\sigma} dt\ \phi(t)\ , 
\end{align}
where $\quad \mu \equiv nq$, $\sigma \equiv \sqrt{nq(1-q)}$ and $\phi(t) \equiv \frac{e^{-t^2/2}}{\sqrt{2\pi}}$. Note that $\sigma$ and $\mu$ depend on $n$, and also on $s$ via $q(s)$. The parameters $l$ and $u$ are specified by the problem Hamiltonian, and are therefore allowed to depend on $n$ as long as $l(n)<u(n)<n$ is satisfied for all $n$.

Let us define:
\beq
\label{eq:intbound}
B(s,n,l(n),u(n)) \equiv \int_{(l(n)-\mu(n,s))/\sigma(n,s)}^{(u(n)-\mu(n,s))/\sigma(n,s)} dt\ \frac{e^{-t^2}/2}{\sqrt{2\pi}}.
\eeq
We seek an upper bound on this function.  We observe that $q(s)$ decreases monotonically from $\frac{1}{2}$ to $0$ as $s$ goes from $0$ to $1$.   Thus, the mean of the Gaussian $\mu(n,s) = n q(s)$ decreases from $\frac{n}{2}$ to $0$. 
Depending on the values of $l(n)$, $u(n)$ and $\mu(n,s)$, we thus have three possibilities: (i) $l(n)<\mu(n,s)<u(n)$, (ii) $\mu(n,s)<l(n)<u(n)$, and (iii) $l(n)<u(n)<\mu(n,s)$. Note that (ii) and (iii) are cases where the integral runs over the tails of the Gaussian and so the integral is exponentially small. We focus on (i), as this induces the maximum values of the integral.
In this case the lower limit of the integral Eq.~\eqref{eq:intbound} is negative, while the upper limit is positive. Thus, the integral runs through the center of the standard Gaussian, and we can upper-bound the value of the integral by the area of the rectangle of width $\frac{u(n)-l(n)}{\sigma(n,s)}$ and height $\frac{1}{\sqrt{2\pi}}$. Hence
\bes
\begin{align}
B(s,n,l(n),u(n)) &\leq \frac{1}{\sqrt{2\pi}} \frac{u(n)-l(n)}{\sigma(n,s)}, \\
&= \frac{1}{\sqrt{2\pi}} \frac{u(n)-l(n)}{\sqrt{\mu(n) (1-q(s))}}, \\
&\leq \frac{1}{\sqrt{2\pi}} \frac{u(n)-l(n)}{\sqrt{l(n) (1-q(s))}},
\end{align}
\ees
where we have used the fact that $l(n) < \mu(n,s) = nq(s)$.

Thus, we obtain the bound:
\beq
\tilde{E}_0 - E_0 \leq \mathcal{O}\left(h \frac{u-l}{\sqrt{l}}\right). \label{eq:bound}
\eeq
\end{proof}

\begin{lemma}[\cite{Reichardt:2004}] 
\label{lem:spec} 
If $\tilde{H} - H$ is non-negative, then the spectrum of $\tilde{H}$ lies above the spectrum of $H$. That is, $\tilde{E}_j \geq E_j$ for all $j$, where $\tilde{E}_j$ and $E_j$ denote the $j$th largest eigenvalue of $\tilde{H}$ and $H$, respectively.
\end{lemma}

This can be proved by a straightforward application of the Courant-Fischer min-max theorem (see, for example, Ref.~\cite{horn2012matrix}).
%
%

Combining these lemmas results in the desired bound on the gap:
\bes
\begin{align}
\text{Gap}[\tilde{H}(s)] &= \tilde{E}_1 - \tilde{E}_0, \\
&\geq E_1 - \tilde{E}_0,  \label{eq:lem1}\\
&= E_1 - E_0 - (\tilde{E}_0 - E_0), \\
&\geq \Delta - \mathcal{O}\left(h \frac{u-l}{\sqrt{l}}\right) \label{eq:lem2},
\end{align}
\ees
where in Eq.~\eqref{eq:lem1} we used Lemma \ref{lem:spec} and in Eq.~\eqref{eq:lem2}, we used Lemma \ref{lem:bound}.

Now, if we choose a parameter regime for the perturbation such that $h \frac{u-l}{\sqrt{l}} = o(1)$, then the perturbed problem maintains a constant gap. For example, if $l = \Theta(n)$ and $h(u-l) = \mathcal{O}(n^{1/2-\epsilon})$, for any $\epsilon > 0$, then the gap is constant as $n \to \infty$.

\section{Adiabatic scaling} \label{app:quantumscaling}
%
In order to study the adiabatic scaling, we consider the minimum time $\tau_{0}$ required to reach the ground state with some probability $p_{\mathrm{ThC}}$, where we choose $p_{\mathrm{ThC}}$ to ensure that we are exploring a regime close to adiabaticity for QA.   We call this benchmark metric the ``threshold criterion,'' and set  $p_{\mathrm{ThC}}=0.9$.  As seen in Fig.~\ref{fig:ThC}, QA scales polynomially, approximately as $n^{0.5}$.  It is also clear that the adiabatic criterion given by Eq.~\eqref{eqt:AdC} provides an excellent proxy for the scaling of QA. 

In light of a spate of recent negative results concerning the possibility of an advantage of SQA over SA (e.g., Ref.~\cite{Heim:2014jf}), it is remarkable, and of independent interest, that SQA scales better than SA for the plateau problem.
\begin{figure}[h] 
\includegraphics[width=0.8 \columnwidth]{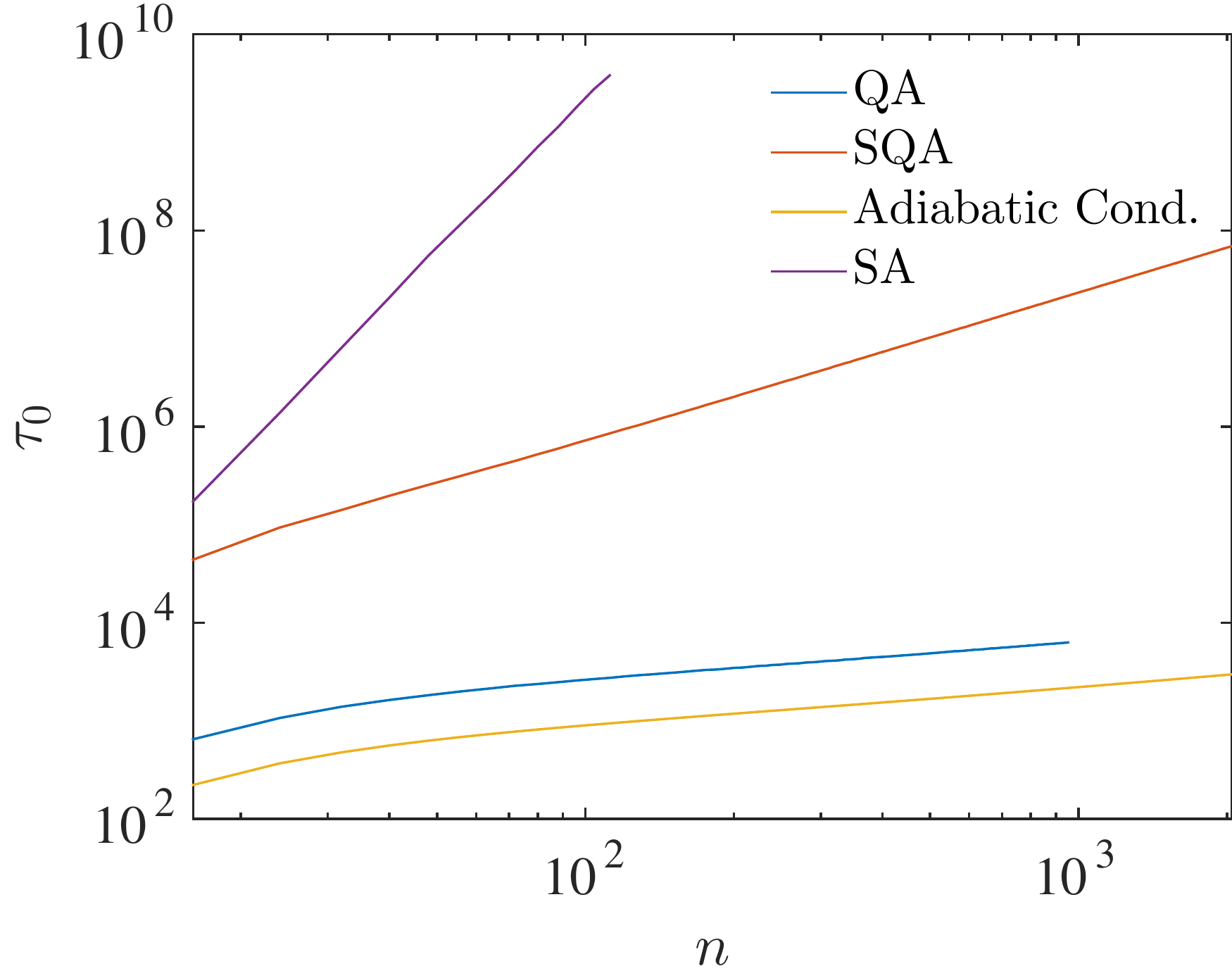}    
 \caption{Log-log plot of the scaling of the time to reach a success probability of $0.9$, as a function of system size $n$ and $u =6$, for QA, SQA ($\beta = 30$, $N_{\tau} = 64$) and SA ($\beta_f = 20$).  The time for SQA and SA is measured in single-spin updates.  We also show the scaling of the adiabatic condition as defined in Eq.~\eqref{eqt:AdC} since it shows the same scaling as QA but can be calculated for larger spin systems. QA and the adiabatic condition scale approximately as $n^{0.5}$.  SQA scales more favorably ($\sim n^{1.5}$) than SA ($\sim n^{5}$).} 
\label{fig:ThC} 
\end{figure}
%

\section{Analysis of simulated annealing using random spin selection} \label{app:SArandom}
%
Here we analyze SA for the plateau and the Hamming weight problems. We consider a version of SA with random spin-selection as the rule that generates candidates for Metropolis updates.

An example of the plateau is illustrated in Fig. \ref{fig:Plateau}, with perturbation applied between strings of Hamming weight $3$ and $8$.  Suppose we start from a random bit-string.  For large $n$, with very high probability, we will start at a bit-string with Hamming weight close to $n/2$.  The plateau may be to the left or to the right of $n/2$; if the plateau is to the right, then most likely the random walker will not encounter it and fall quickly to the ground state in at most $\mathcal{O}(n^2)$ steps (see a few paragraphs below for a derivation).

\begin{figure}[htbp] 
   \centering
   \includegraphics[width=0.8\columnwidth]{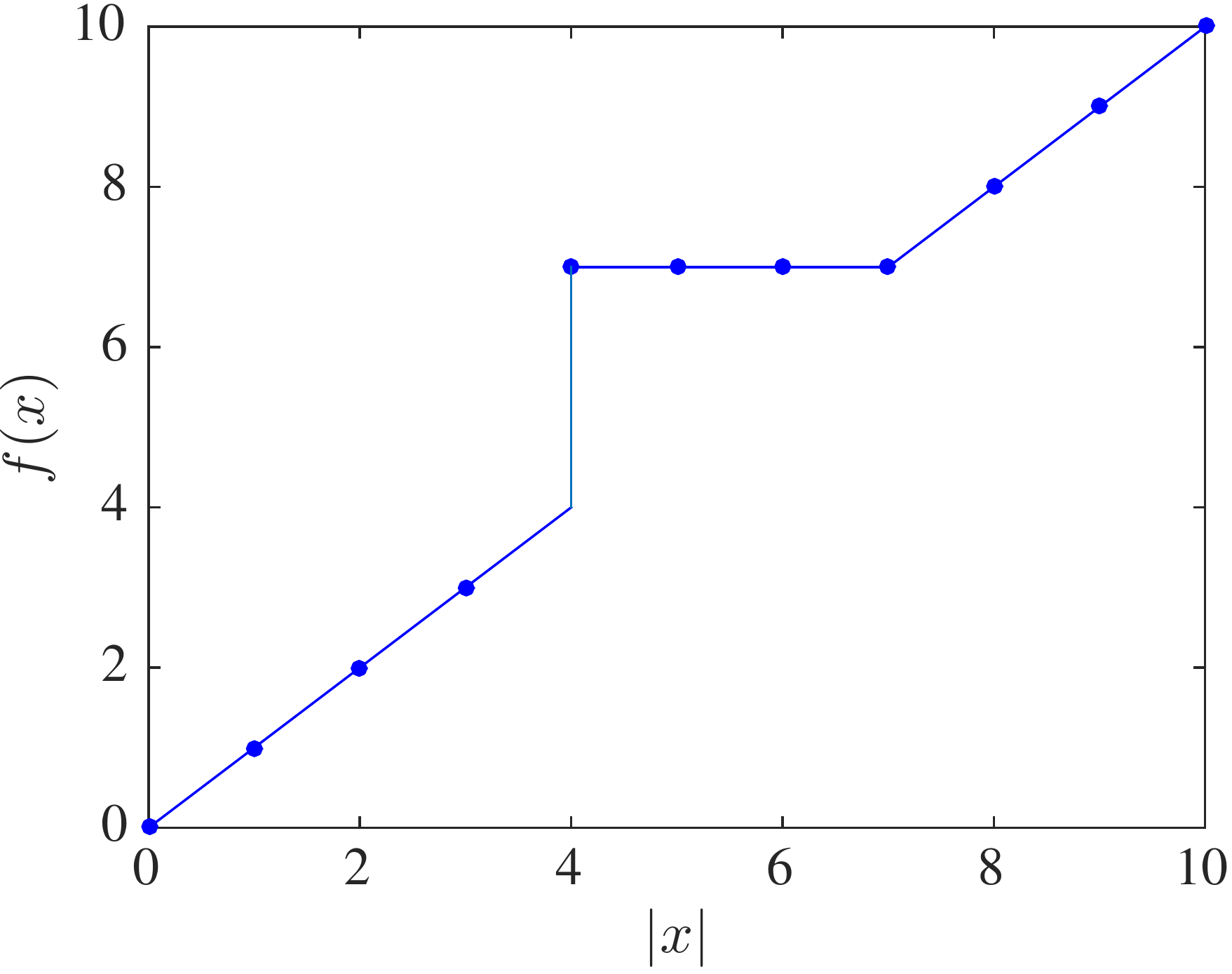} 
   \caption{$l=3,u=8$}
   \label{fig:Plateau}
\end{figure}

Thus, the interesting case is when the random walker arrives at the plateau from the right.  In this case, how much time would it take, typically, for the walker to fall off the left edge? It is intuitively clear that traversing the plateau will be the dominant contribution to the time taken to reach the ground state, as after that the random walker can easily walk down the potential. As we show later below, this time can be at most $\mathcal{O}(n^2)$ (ignoring transitions which take it back onto the plateau) for an inverse temperature that scales as $\beta = \Omega(\log n)$.

To evaluate the time to fall off the plateau, let us model the situation as follows. First, note that the perturbation is applied on strings of Hamming weight $l+1,l+2,\dots,u-1$, so the width of the plateau is $w = u-l-1$. Consider a random walk on a line of $w+1$ nodes labelled $0,1,\dots w$. Node $i$ represents the set of bit strings with Hamming weight $l+i$, with $0\leq i \leq w$. We assume that the random walker starts at node $w$, as it is falling onto the right edge of the plateau. Only nearest-neighbor moves are allowed and the walk terminates if the walker reaches node $0$.

Our model will estimate a shorter than actual time to fall off the left edge, because in the actual PHWO problem one can also go back up the slope on the right, and in addition we disallow transitions from strings of Hamming weight $l$ to $l+1$. This is justified because the Metropolis rule exponentially (in $\beta$) suppresses these transitions.

The transition probabilities $p_{i \to j}$ for this problem can be written as a $(w+1) \times (w+1)$ row-stochastic matrix $p_{ij} = p_{i \to j}$. $p$ is a tridiagonal matrix with zeroes on the diagonal, except at $p_{00}$ and $p_{ww}$.  First consider $1\leq i \leq w-1$. If the walker is at node $i$, then its Hamming weight is $l+i$. Thus walker will move to $i+1$ (which has Hamming weight $l+i+1$) with probability $\frac{n-(l+i)}{n}$ (the chance that the bit picked had the value $0$). Now consider, $1\leq i \leq w$ the Hamming weight will decrease to $l+i-1$ with probability $\frac{l+i}{n}$ (the chance that the bit picked had the value $1$). Combining this with the fact that a walker at node $0$ stays put, we can write:
\bes
\begin{align}
b_i &\equiv p_{i\to i}=\begin{cases} 1 \text{  if  } i=0 \\
0 \text{  if  } 1\leq i \leq (w-1) \\ 
1-\frac{l+w}{n} \text{  if  } i=w 
\end{cases},\\		  
c_i &\equiv p_{i-1\to i} = \begin{cases} 0 \text{  if  } i=1\\ 1-\frac{l+i-1}{n} \text{  if  } i=2,\dots,w \end{cases}, \\
a_i &\equiv p_{i\to i-1} = \frac{l+i}{n} \text{  if  } i=1,2,\dots,w.
\end{align}
\ees
Let $X(t)$ be the position of the random walker at time-step $t$. The random variable measuring the number of steps taken by the random walker starting from node $r$ would to reach node $s$ for the first time is
\beq
\tau_{r,s} \equiv \text{min} \{ t>0: X(t)=s,X(t-1)\neq s | X(0) =r\}\ .
\eeq
The quantity we are after is $\mathbb{E}\tau_{w,0}$, the expectation value of the random variable $\tau_{w,0}$, i.e., the mean time taken by the random walker to fall off the plateau. Since only nearest neighbor moves are allowed we have
\beq 
\mathbb{E}\tau_{w,0} =\sum_{r=1}^w \mathbb{E}\tau_{r,r-1}\ .
\eeq
Stefanov~\cite{stefanov1995mean} (see also Ref.~\cite{krafft1993mean}) has shown that
\beq \label{eq:stefanovformula}
\mathbb{E}\tau_{r,r-1} = \frac{1}{a_r} \left( 1 + \sum_{s=r+1}^{w} \prod_{t=r+1}^{s} \frac{c_t}{a_t} \right),
\eeq
where $c_{w+1}\equiv 0$. Evaluating the sum term by term: 
\bes
\label{eq:62}
\begin{align}
\mathbb{E}\tau_{w,w-1} &= \frac{n}{l+w}, \\
\mathbb{E}\tau_{w-1,w-2} &= \frac{n}{l+w-1} \left[1 + \frac{n-(l+w-1)}{l+w} \right], \\
\mathbb{E}\tau_{w-2,w-3} &= \frac{n}{l+w-2} \left[1 + \frac{n-(l+w-2)}{l+w-1} \right. \nonumber \\
& \left. + \frac{n-(l+w-2)}{l+w-1} \times \frac{n-(l+w-1)}{l+w}\right], \\
\vdots \nonumber \\
\mathbb{E}\tau_{w-k,w-k-1} &= \frac{n}{l+w-k} \left[1+ \frac{n-(l+w-k)}{l+w-(k-1)}+\dots \right. \nonumber \\
&  +\frac{n-(l+w-k)}{l+w-(k-1)}\times\cdots  \nonumber \\ 
& \left. \times\frac{n-(l+w-2)}{l+w-1} \times \frac{n-(l+w-1)}{l+w}\right]. 
\end{align}
\ees
Now consider the following cases:
\begin{enumerate}
\item $l,u = \mathcal{O}(1)$: Here, using the fact that $k=\mathcal{O}(w)=\mathcal{O}(1)$, we conclude that $\mathbb{E}\tau_{w-k,w-k-1} = \mathcal{O}(n^{k+1})$. Since the leading order term is $\mathbb{E}\tau_{w-(w-1),w-w}=\mathbb{E}\tau_{1,0}$, the time to fall off the plateau is $\mathcal{O}(n^w) = \mathcal{O}(n^{u-l-1}).$
\item For Reichardt's bound to give a constant lower-bound to the quantum problem, we need $u=l+o(l^{1/4})$. Since at most we can have $l = \mathcal{O}(n)$, we can conclude $\mathbb{E}\tau_{w-k,w-k-1} = \mathcal{O}\left(\frac{n}{l}\right)^{k+1}$. Therefore, the time to fall-off becomes $\mathbb{E}\tau_{w,0} = \mathcal{O}\left(w(\frac{n}{l})^w\right)$. 
\begin{itemize}
\item If $l = \Theta(n)$ and $w = \mathcal{O}(1)$, we can see that $\mathbb{E}\tau_{w,0} = \mathcal{O}(1)$, which is a constant time scaling.
\item If $l = \Theta(n)$ and $w = \mathcal{O}(n^a)$, where $0<a<1/4$, then $\mathbb{E}\tau_{w,0} = \mathcal{O}(n^a \mathcal{O}(1)^{n^a})$, which is super-polynomial.
\item More generally, if $l = \mathcal{O}(n^b)$, with $b \leq 1$ and $w = \mathcal{O}(n^a)$, where $0\leq a < b/4$, then we get the scaling $\mathbb{E}\tau_{w,0} = \mathcal{O}(n^{a}\mathcal{O}(n^{1-b})^{n^a})$
\end{itemize}
\end{enumerate}

\textit{Analysis of SA for plain Hamming weight} ---
Let us analyze the behavior of a fixed temperature, i.e., there is no annealing schedule, simulated annealing on the plain Hamming weight problem. Here the transition probabilities are:
\bes
\begin{align}
c_i &\equiv p_{i-1\to i} = \frac{n-i+1}{n} e^{-\beta} \ , \\
a_i &\equiv p_{i\to i-1} =  \frac{i}{n} \ ,
\end{align}
\ees
with $i=1,2,\dots,n$ denoting strings of Hamming weight $i$, and $\beta$ is the inverse temperature. Using the Stefanov formula~\eqref{eq:stefanovformula}, we can write (after much simplification):
\beq
\mathbb{E}\tau_{n-k,n-k-1} = \frac{n}{n-k} \binom{n}{k}^{-1} \sum_{l=0}^{k} e^{-l \beta} \binom{n}{k-l} \ .
\eeq
Thus, 
\beq \label{eqt:hamwttime}
\mathbb{E}\tau_{n,0} = \sum_{k=0}^{n-1} \frac{n}{n-k} \binom{n}{k}^{-1} \sum_{l=0}^{k} e^{-l \beta} \binom{n}{k-l} \ .
\eeq
This is the worst-case scenario as we are assuming that we start from the string of Hamming weight $n$, which is the farthest from the all-zeros string. Note that if we start from a random spin configuration, then with overwhelming probability, we will pick a string with Hamming weight close to $n/2$. Thus, most probably, $\mathbb{E}\tau_{n/2,0}$ will be the time to hit the ground state. We can write this as:
\beq
\mathbb{E}\tau_{n/2,0} = \sum_{k=n/2}^{n-1} \frac{n}{n-k} \binom{n}{k}^{-1} \sum_{l=0}^{k} e^{-l \beta} \binom{n}{k-l}.
\eeq

We first show that $\beta = \mathcal{O}(1)$ will lead to an exponential time to hit the ground state. We show this by showing that $\mathbb{E}\tau_{1,0}$ is exponential if $\beta = \mathcal{O}(1)$. To this end,
\bes
\begin{align}
\mathbb{E}\tau_{1,0} &= \mathbb{E}\tau_{n-(n-1),n-n} \\
&=\sum_{l=0}^{n-1} e^{-l \beta} \binom{n}{n-1-l} \\
&=e^\beta\left[(e^{-\beta}+1)^n-1\right],
\end{align}
\ees
which is clearly exponential in $n$ if $\beta= \mathcal{O}(1)$. Now let us suppose we have $\beta = \log n$, i.e. we decrease the temperature inverse logarithmically in system size. In this case,
\begin{align}
\mathbb{E}\tau_{1,0} = n \left[\left(1 + \frac{1}{n}\right)^n-1\right] \leq n (e - 1) = \mathcal{O}(n)\ .
\end{align}
Now it is intuitively clear that $\mathbb{E}\tau_{1,0} > \mathbb{E}\tau_{r,r-1}$ for all $r>1$, which implies that $n \mathbb{E}\tau_{1,0} \geq \mathbb{E}\tau_{n,0}  $. Thus, if $\beta = \log n$, then $\mathbb{E}\tau_{n,0} = \mathcal{O}(n^2)$ at worst.

To obtain a lower-bound on the performance of the algorithm, we take $\beta \to \infty$. Thus,  for each $k$ in Eq.\eqref{eqt:hamwttime}, only the $l=0$ term will survive. Thus,
\begin{align}
\lim_{\beta \to \infty} \mathbb{E}\tau_{n,0} &= \sum_{k=0}^{n-1} \frac{n}{n-k} = n \sum_{i=1}^n \frac{1}{i} \approx n (\log n + \gamma)\ ,
\end{align}
for large $n$, with $\gamma$ as the Euler-Mascheroni constant. So the scaling here is $\mathcal{O}(n \log n)$. This is the best possible performance for single-spin update SA with random spin-selection on the plain Hamming weight problem. Therefore, if $\beta = \Omega(\log n)$, the scaling will be between $\mathcal{O}(n \log n)$ and $\mathcal{O}(n^2)$.

To conclude, let us make a few remarks on the three different benchmarking metrics used in this work: (i)~TTS$_\mathrm{opt}$, (ii)~the mean time to first hit the ground state ($\mathbb{E}\tau_{n/2,0}$ or $\mathbb{E}\tau_{n,0}$), and (iii)~the time to cross a particular fixed threshold probability (typically high, say $p_{\mathrm{ThC}}=0.9$) of finding the ground state ($\tau_{\mathrm{ThC}}$). In order to compare the three, we would need to find  TTS($t_f = \mathbb{E}\tau_{n/2,0}$) and TTS($t_f=\tau_{\mathrm{ThC}}$).  By definition, TTS$_\mathrm{opt}$, will be the smallest of the three. Further note that typically $p_{\mathrm{GS}}(t_f = \mathbb{E}\tau_{n/2,0}) < p_{\mathrm{ThC}}$. This implies that if $t_f = \mathbb{E}\tau_{n/2,0}$, then the algorithm would need to be repeated more times than if $t_f = \tau_{\mathrm{ThC}}$ to obtain the same confidence that we have seen the ground state at least once. Note that either could have the smaller TTS, depending on the problem at hand.

We remark that the analysis performed above for random spin-selected SA with the complexity metric as $\mathbb{E}\tau_{n/2,0}$, captures extremely well the numerically obtained scaling of sequential spin-selected SA with the complexity metric as $\tau_{\mathrm{ThC}}$.

\section{Other PHWO problems} 
\label{app:DiabaticScaling}

\begin{figure}[t] 
   \subfigure[]{\includegraphics[width=0.8\columnwidth]{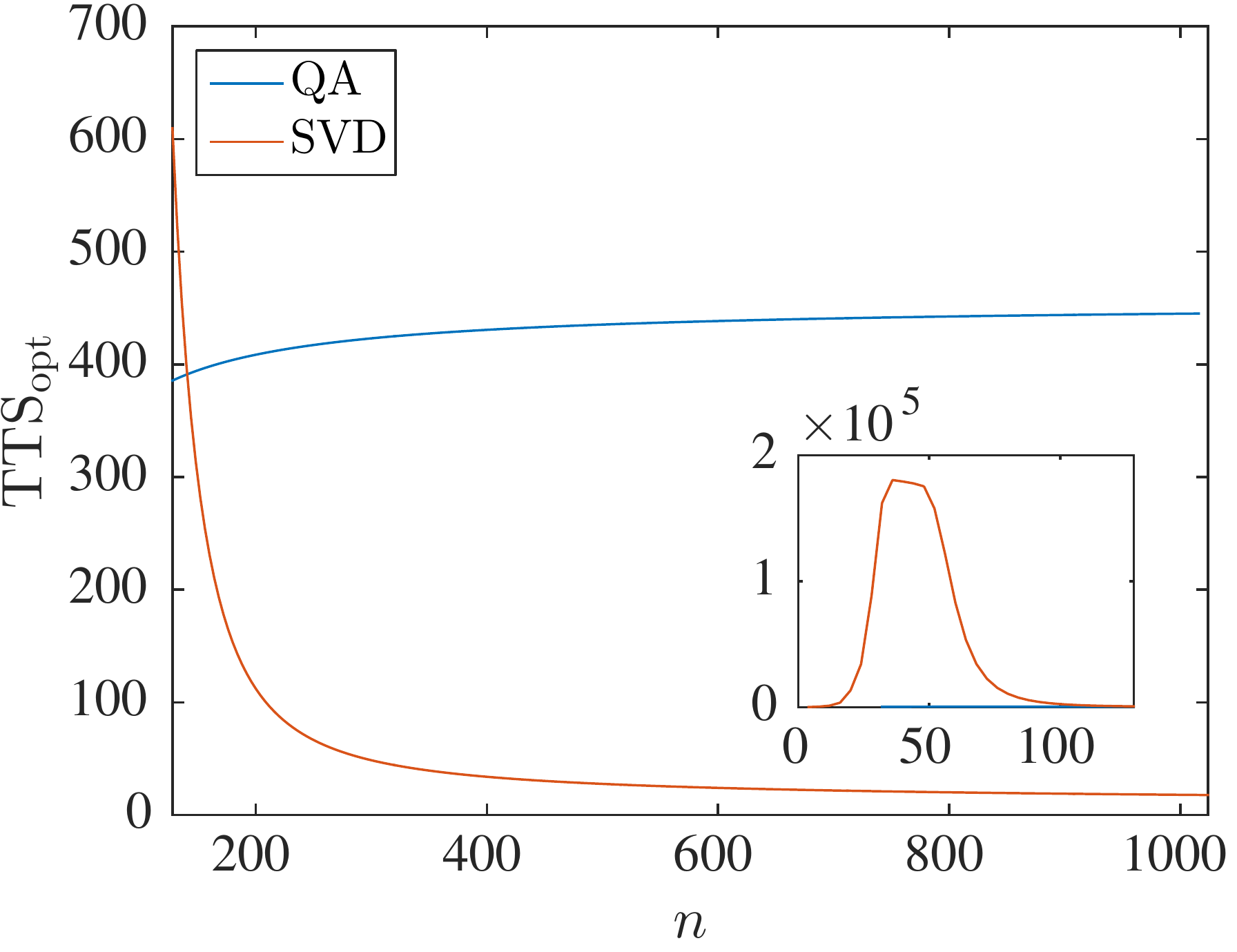}}
   \subfigure[]{\includegraphics[width=0.8\columnwidth]{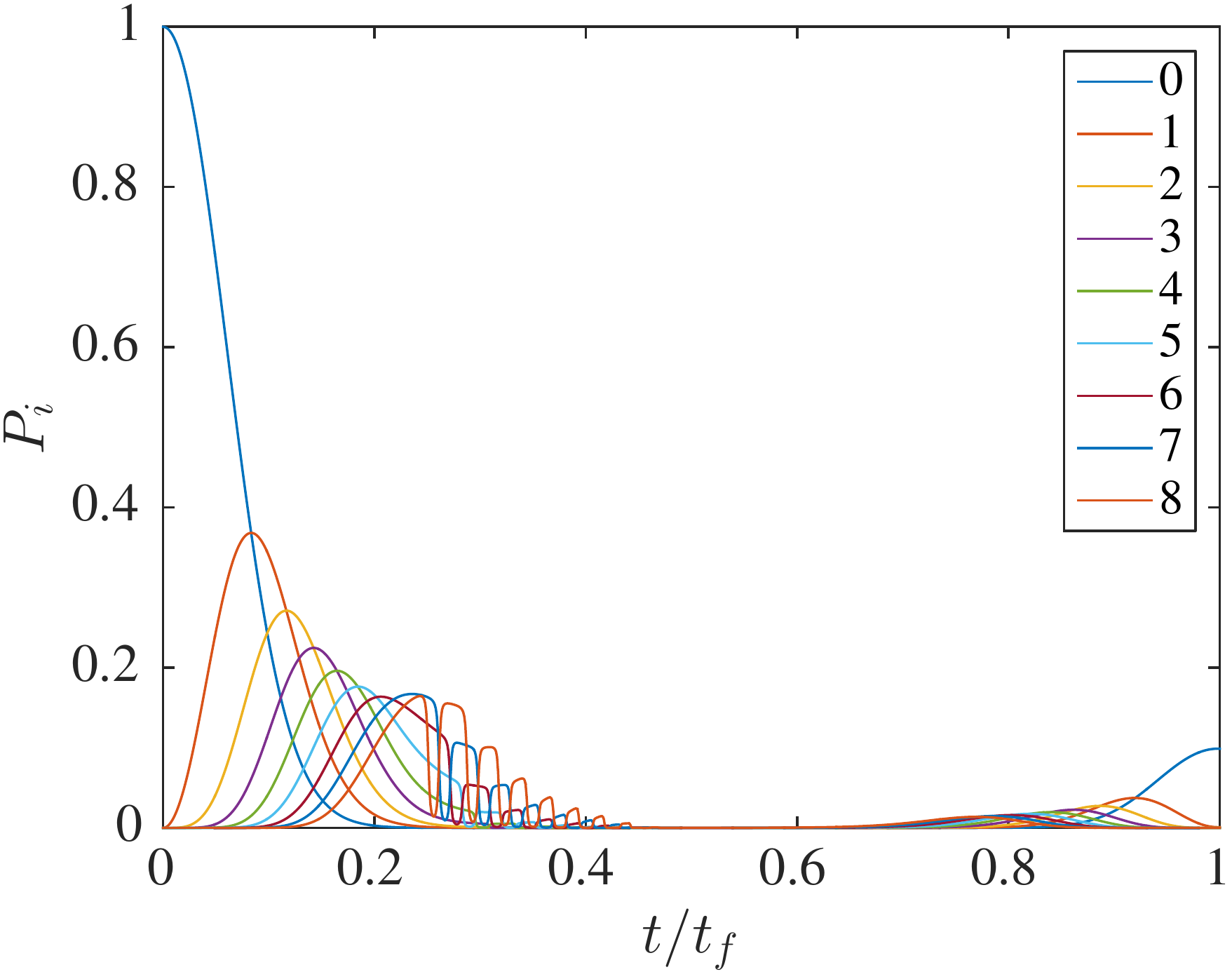}}
   \caption{(a) The optimal TTS for the spike problem \cite{Farhi-spike-problem}.  Inset: the optimal TTS for small problem sizes, where we observe SVD at first scaling poorly. However, as $n$ grows, this difficulty vanishes and it quickly beats QA. (b) We observe similar diabatic transitions for this problem (shown is $n=512$ and $t_f = 9.85$) as we observed for the plateau [Fig.~\ref{fig:QA_EnergyOverlap}], although here the success probability appears to saturate faster than for the plateau problem.  
}
\label{fig:Spike_TTSopt}
\end{figure}   

\begin{figure}[t] 
   \subfigure[]{\includegraphics[width=0.8\columnwidth]{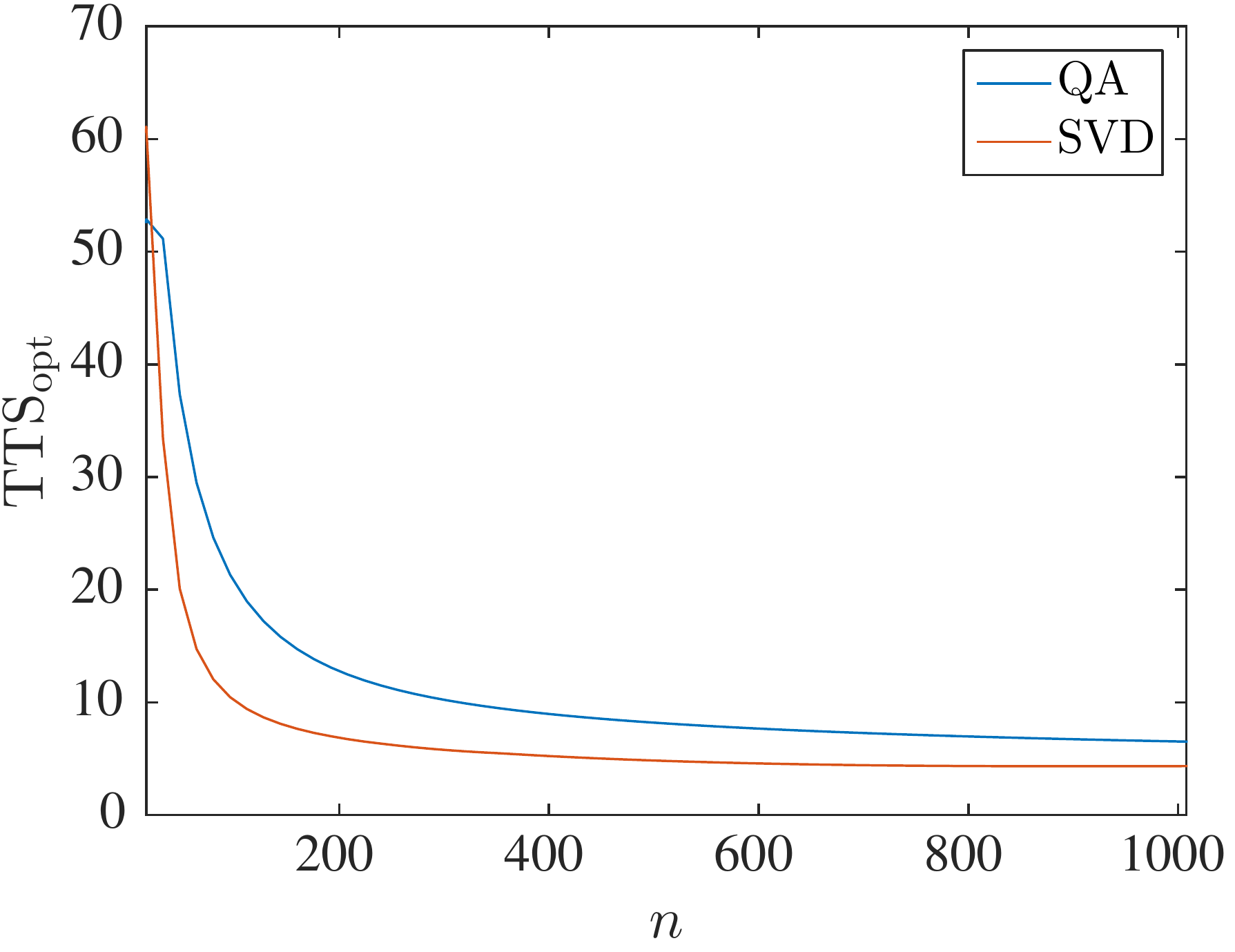}}
   \subfigure[]{\includegraphics[width=0.8\columnwidth]{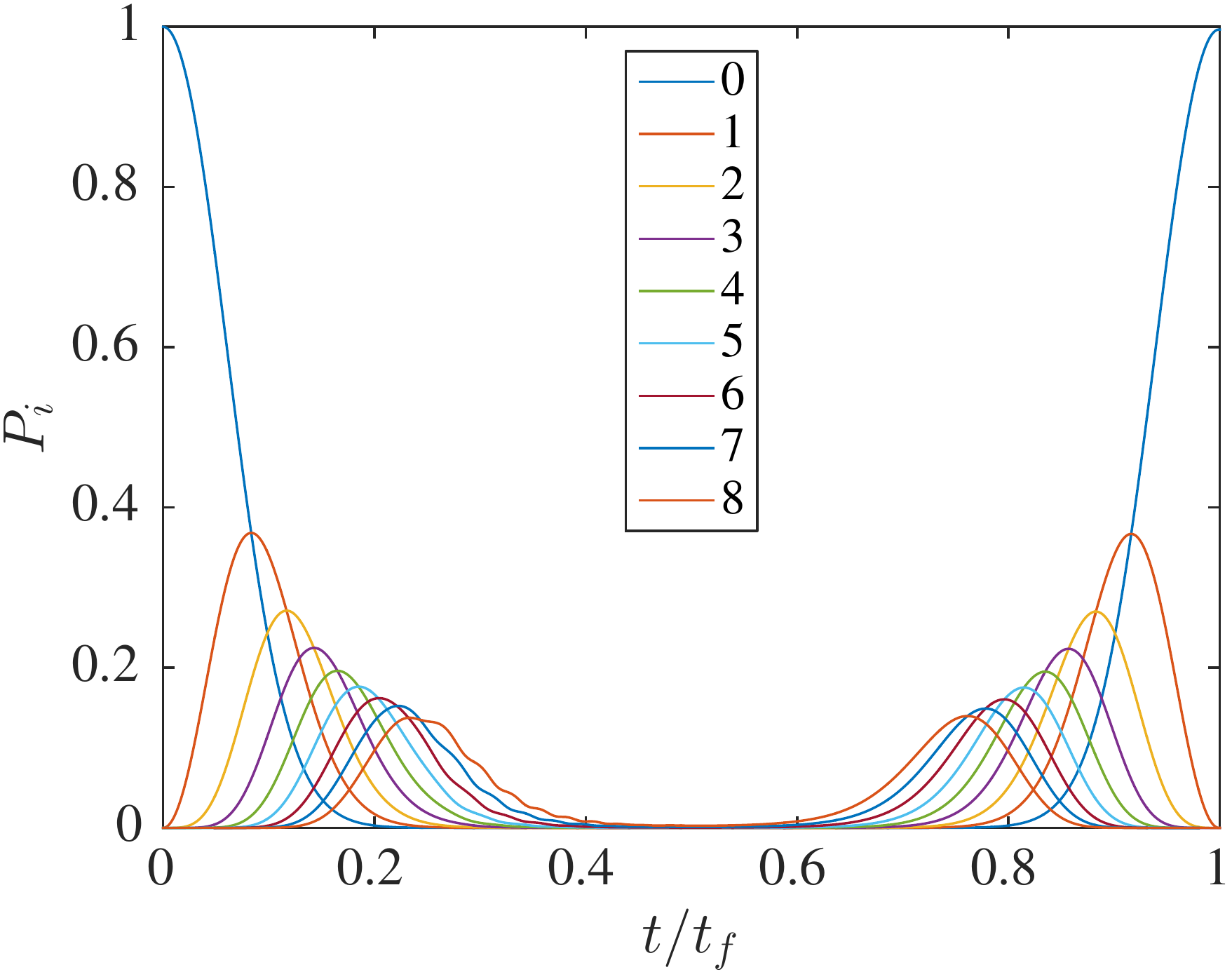}}
  \caption{(a)  The optimal TTS for the plateau with $l(n)=n/4$ and $u(n)=l(n)+6$.  As can be seen, SVD is much better than QA for these problem sizes. (b) We observe similar diabatic cascades for this problem (shown is $n=512$ and $t_f = 10$) as we observed for the plateau problem [Fig.~\ref{fig:QA_EnergyOverlap}]. 
}
\label{fig:movingplatTTSopt}
\end{figure}

In this section we consider several other versions of the PHWO problem.
The first two examples exhibit diabatic cascades, while the last does not. 
\begin{enumerate}
\item The ``spike" problem studied by Farhi \textit{et al}. \cite{Farhi-spike-problem} has the following cost function:
\beq
f(x) = \begin{cases} n, & \text{if  } |x| = \frac{n}{4},\\
\abs{x}, & \text{elsewhere} \end{cases} \ .
\eeq
This too is a problem designed explicitly to stymie SA (in Ref.~\cite{Farhi-spike-problem} it is argued that SA will take exponential time) and has polynomially decreasing gap $~\mathcal{O}(n^{-1/2})$ (and thus will have some polynomial run-time in the adiabatic regime). In Fig.~\ref{fig:Spike_TTSopt} we show that this problem too shows diabatic cascades and a corresponding outperformance by SVD. 

\item We pick an instance of the plateau with $l=n/4$ [i.e., $\mathcal{O}(n)$] and $u=l+6$ [i.e., $l+\mathcal{O}(1)$]. This problem has a constant lower-bound for QA  by Reichardt's theorem [see Eq.~\eqref{eqt:Reichbound}], and SA is able to solve it in constant time [recall the discussion below Eq.~\eqref{eq:62}]. 
In Fig.~\ref{fig:movingplatTTSopt}, we see that this problem too exhibits diabatic transitions for QA and an advantage for SVD. For this problem, as we show in SM-\ref{app:asymptoticVSC}, the semiclassical effective potential asymptotically becomes identical to the unperturbed Hamming weight problem, which explains why the TTS$_\mathrm{opt}$ for this is decreasing: the TTS$_\mathrm{opt}$ for the (plain) Hamming weight problem is constant.


\item Consider the following class of PHWO problems, introduced in Ref.~\cite{vanDam:01}:
\beq
f(x) = \begin{cases} p(|x|), &  |x| > \left( \frac{1}{2}+\epsilon \right)\\
\abs{x}, & |x| \leq \left( \frac{1}{2}+\epsilon \right) \end{cases} \ ,
\eeq
where $\epsilon>0$ and $p(\cdot)$ is a decreasing function which attains the global minimum, $-1$, in the $|x| > \left( \frac{1}{2}+\epsilon \right)$ region. Ref.~\cite{vanDam:01} proved that this class of problems has an exponentially decreasing gap, and therefore the adiabatic algorithm would take exponentially long to find the ground state. We have considered the following instance of this class:
\beq
f(x) = \begin{cases} -1, &  |x| = n,\\
\abs{x}, & \text{otherwise}  \end{cases} \ .
\eeq
In this case, we did not observe the diabatic transition phenomenon (not shown), i.e., the optimal TTS is achieved by evolving adiabatically and remaining in the ground state. Thus the diabatic transition phenomenon does not persist for all PHWO problems. 
\end{enumerate}

\section{Asymptotic behavior of semiclassical effective potentials}
\label{app:asymptoticVSC}
Here we analyze the behavior of the (symmetric) effective potential we found in Eq.~\eqref{eqt:symvsc} and write down here in simplified form:
\begin{align}
V_{\mathrm{SC}}(\theta, \varphi, s) &\equiv \bra{\theta,\varphi} H(s) \ket{\theta,\varphi}, \notag \\
&= \frac{n}{2}(1-s)(1-\sin\theta\cos\varphi) \notag \\
&\quad +s\sum_{w=0}^n f(w) \binom{n}{w} p(\theta)^k (1-p(\theta))^{n-k},
\end{align}
where $p(\theta)\equiv \sin^2\left(\frac{\theta}{2}\right)$. We take $f(w)$ to be a PHWO  Hamiltonian of the form of Eq.~\eqref{eq:pertham}. We can write the plateau's effective potential as:
\beq
V_{\mathrm{SC}}^{\mathrm{pert}} = V_{\mathrm{SC}}^{{\mathrm{unpert}}}+s\sum_{l<k<u} f(k) \binom{n}{k} p(\theta)^k (1-p(\theta))^{n-k}.
\eeq
Note the resemblance between the perturbation in the above equation and the term that appears in Reichardt's lower-bound [see Eq.~\eqref{eq:hk}]. The only difference is that we have replaced $q(s)$ with $p(\theta)$. Now, if we trace through the arguments deriving the lower-bound on the gap, we see that the same holds for the perturbation term here. In particular:
\beq
\sum_{l<k<u} f(k) \binom{n}{k} p(\theta)^k (1-p(\theta))^{n-k} = \mathcal{O}\left(h \frac{u-l}{\sqrt{l}}\right).
\eeq
Therefore, when $l,u = \mathcal{O}(1)$, the semiclassical effective potential asymptotically maintains a perturbation relative to the unperturbed problem. On the other hand, for the cases $l = \mathcal{O}(n), u = l + \mathcal{O}(1)$ and $l = \mathcal{O}(n), u = l  + \mathcal{O}(n^{1/4 - \epsilon})$,  the perturbation to the semiclassical effective potential vanishes asymptotically. This shows that the effective potential leads to equivalent conclusions about computational hardness as the gap analysis.

\section{Behavior of the average Hamming weight on the classical Gibbs state} \label{app:Gibbs}

\begin{figure*}[t] 
   \subfigure[]{\includegraphics[width=0.3\textwidth]{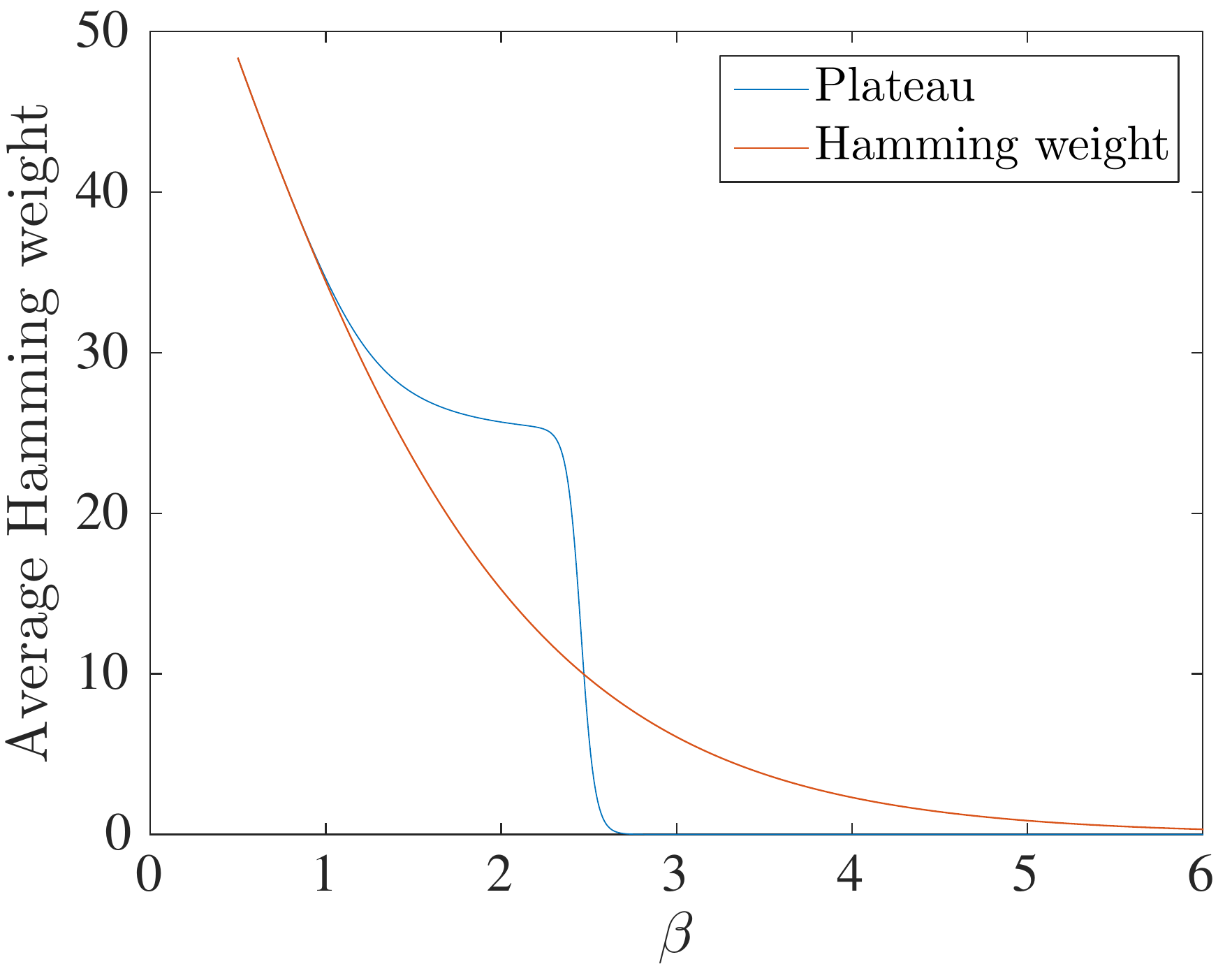} \label{fig:AvgHamWtPlatHam}}
   \subfigure[]{\includegraphics[width=0.3\textwidth]{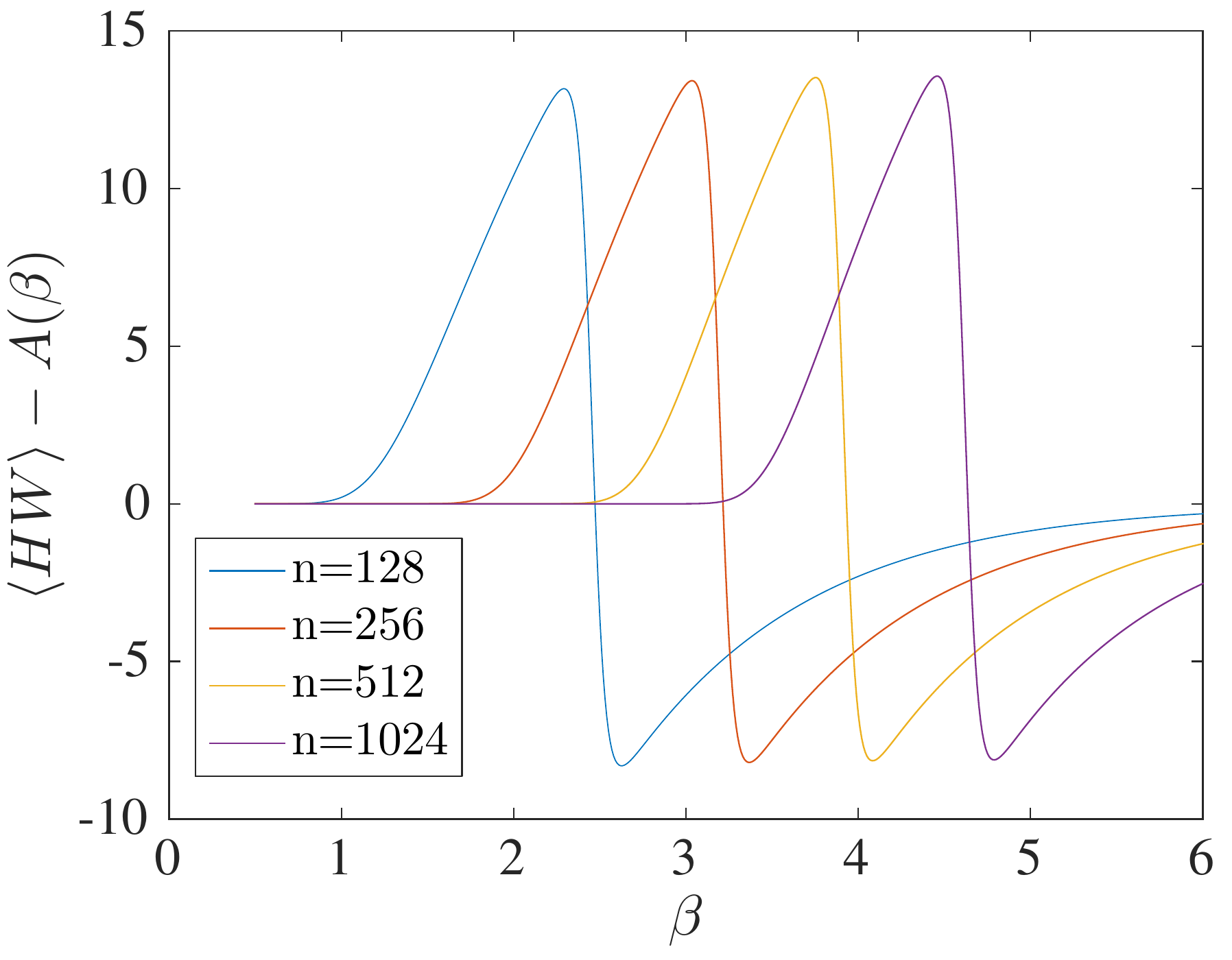}\label{fig:AvgHamWtStatPlat}}
    \subfigure[]{\includegraphics[width=0.3\textwidth]{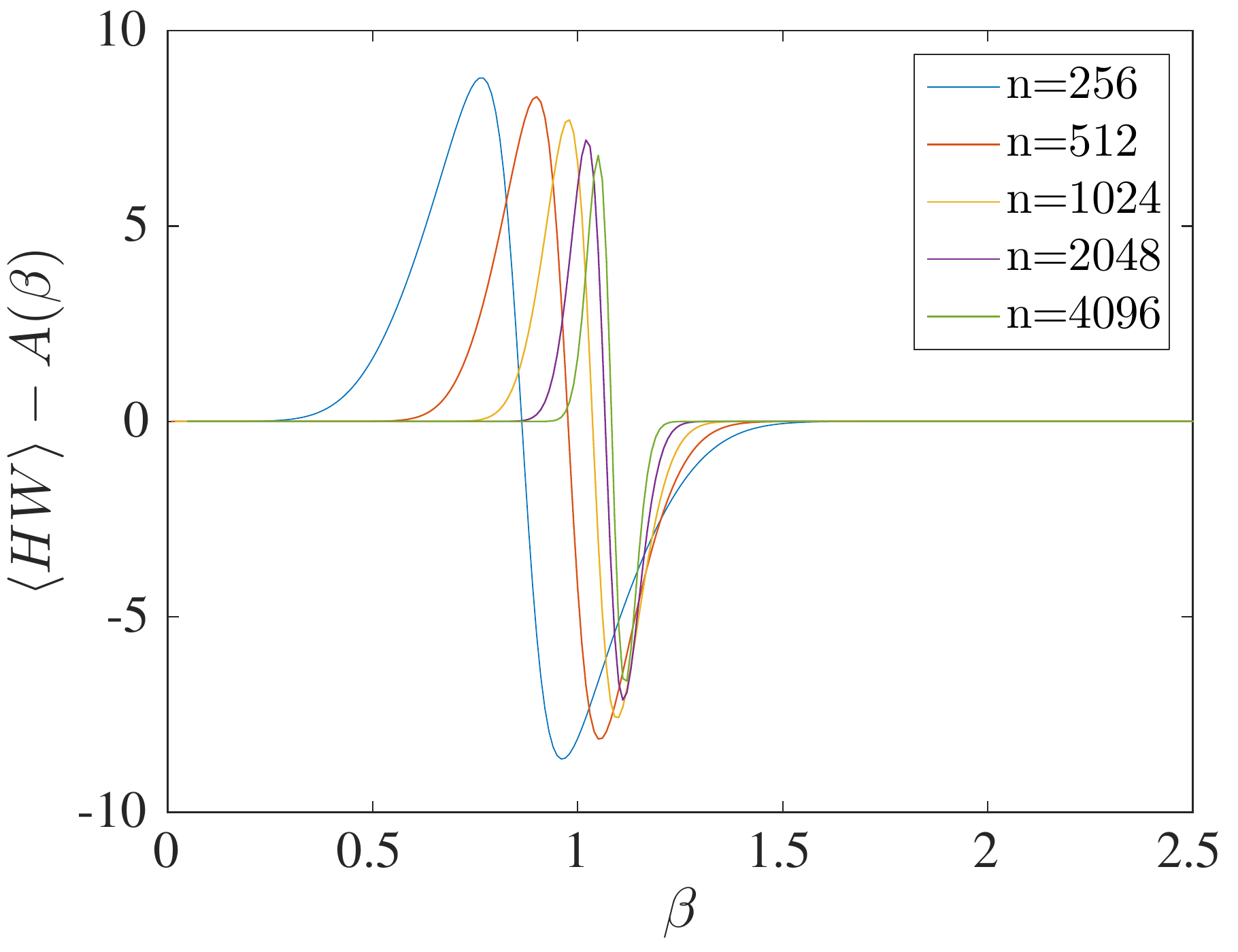}\label{fig:AvgHamWtMovPlat}}
   \caption{(a) $\langle \mathrm{HW} \rangle$ in the Gibbs state of the plain Hamming weight function and the plateau function with $l=0$ and $u=26$ for $n=128$. The two functions agree  closely except in the region of the ``drop." (b) The ``signal" $\wich{\mathrm{HW}}-A(\beta)$ for $l=0,u=26$, for $n=128,256,512,1024$. The same sharp drop is seen for all $n$. (c) $\wich{\mathrm{HW}}-A(\beta)$ for the case $l(n)=n/4,u(n)=l(n)+26$, for $n=256,512,1024,2048,3200,4096$. Here the drop is slowly decreasing with $n$.}
\end{figure*}   

In this section we expand on the behavior of the average Hamming weight $\wich{\mathrm{HW}}$ for different cases of the plateau problem. 

In Fig.~\ref{fig:GSGibbs}, we plotted $\wich{\mathrm{HW}}$ for the classical Gibbs state as a function of the inverse temperature, $\beta$.  We interpreted the sharp drops in this quantity as a sign that the problem becomes hard for SA. To understand this better we can consider these sharp drops as modifications of the smooth behavior [see Fig.~\ref{fig:AvgHamWtPlatHam}] of the \emph{plain} Hamming weight function, i.e., Eq.~\eqref{eq:pertham} with $p(\abs{x})=0$. For this case:

\beq
A(\beta) \equiv \wich{\mathrm{HW}}_\mathrm{HamWt} = \frac{n e^{-\beta}}{1+e^{-\beta}}.
\eeq
In order to study just the ``drop," we subtract $A(\beta)$ as the ``background," and focus our attention on the ``signal," which is the sharp change.

We consider the following two varieties of the plateau:
\begin{enumerate}
 \item  $l,u = \mathcal{O}(1)$. This was the case studied in the main text, and we  proved above
 that in this case SA requires polynomial time. As we can see from Fig.~\ref{fig:AvgHamWtStatPlat}, the sharpness of the drop remains constant with increasing $n$.  This is consistent with the problem being hard for SA. 

\item  $l = \mathcal{O}(n)$ and $u = l + \mathcal{O}(1)$. Reichardt's lower bound applies in this case, and we proved above
that this case is solved in time $\mathcal{O}(1)$ by SA. As can be seen in Fig.~\ref{fig:AvgHamWtMovPlat} the sharpness of the drop decreases (albeit slowly) with $n$.  This is consistent with the problem being easy for SA.

\end{enumerate}

To conclude we remark on another case which has constant gap lower-bound by Reichardt's proof. Here, $l = \mathcal{O}(n), u = l  + \mathcal{O}(n^{1/4 - \epsilon})$, $\epsilon>0$. We do not find any dramatic changes in the instantaneous quantum ground state during the evolution as in Fig.~\ref{fig:GSGibbs}, suggesting that multi-qubit tunneling does not play a significant role, hence making it a less relevant problem for our discussion.

\end{document}